\PassOptionsToPackage{svgnames}{xcolor}
\documentclass[a4paper,USenglish,cleveref,autoref,thm-restate]{article}
\usepackage{doi}
\usepackage{fullpage}
\usepackage{amsmath}
\usepackage{amssymb}
\usepackage{amsthm}
\usepackage{xspace}
\usepackage{complexity}
\usepackage{tikz}
\usepackage{todonotes}
\usepackage{extarrows}
\usepackage[linesnumbered,ruled]{algorithm2e}
\usepackage{booktabs}
\usepackage{authblk}
\usepackage{subcaption}

\newtheorem{theorem}{Theorem}
\newtheorem{lemma}[theorem]{Lemma}
\newtheorem{proposition}[theorem]{Proposition}

\newtheorem{corollary}[theorem]{Corollary}
\theoremstyle{definition}
\theoremstyle{remark}

\usetikzlibrary{positioning,calc,decorations.pathmorphing,patterns,shapes}

\newcommand{\yes}{\ensuremath{\mathsf{yes}}\xspace}

\newcommand{\midcolon}{:}
\newcommand{\symdiff}{\mathbin{\triangle}}

\newcommand{\weig}{w} 

\newcommand{\Redset}[1]{#1^{\mathrm{red}}}
\newcommand{\Blueset}[1]{#1^{\mathrm{blue}}}
\newcommand{\Ered}{\Redset{E}}
\newcommand{\Eblue}{\Blueset{E}}
\newcommand{\Ared}{\Redset{A}}
\newcommand{\Ablue}{\Blueset{A}}

\newcommand{\andsc}{\textsc{and}\xspace}
\newcommand{\orsc}{\textsc{or}\xspace}
\newcommand{\andorsc}{\textsc{and/or}\xspace}
\newcommand{\Vand}{V_{\andsc}}
\newcommand{\Vor}{V_{\orsc}}

\newcommand{\indegree}[1]{\rho_{#1}}

\newcommand{\onestep}{\leftrightarrow}
\newcommand{\reachable}{\leftrightsquigarrow}
\newcommand{\reconfseq}[1]{\langle #1 \rangle}

\newcommand{\ini}{\mathrm{ini}}
\newcommand{\tar}{\mathrm{tar}}
\newcommand{\ncl}{\textsc{NCL}\xspace}

\newcommand{\BCSR}{\textsc{BCSR}\xspace}

\newcommand{\oriG}{\hat{G}}
\newcommand{\oriA}{\hat{A}}
\newcommand{\oriV}{\hat{V}}
\newcommand{\oriE}{\hat{E}}

\newcommand{\Vmid}{V_{\textsc{mid}}}

\newcommand{\Eand}{E_{\andsc}}
\newcommand{\Eor}{E_{\orsc}}

\newcommand{\Var}{X} 
\newcommand{\Ecns}{F} 
\newcommand{\Dommap}{D}
\newcommand{\Dom}[1]{\Dommap(#1)}
\newcommand{\Cnsmap}{\mathcal{C}}
\newcommand{\Cns}[1]{\Cnsmap(#1)}
\newcommand{\mapf}{\Gamma}

\newcommand{\doms}{d}
\newcommand{\NB}{p}

\def\ProblemA{{\sc Problem A}\xspace}

\newcommand{\kreconfig}[1][$k$]{\textsc{Reconfiguration}\xspace}

\newcommand{\X}[1]{}

\newcommand{\YES}{\ensuremath{\mathsf{yes}}\xspace}
\newcommand{\NO}{\ensuremath{\mathsf{no}}\xspace}

\newcommand{\ctoc}{\textsc{C2C}\xspace}
\newcommand{\ctoe}{\textsc{C2E}\xspace}
\newcommand{\kernelmark}{\ensuremath{\dagger}}

\newtheorem{myclaim}{Claim}
\newtheorem{reduction}{Reduction rule}

\usetikzlibrary{decorations.pathmorphing,calc,fit}
\usetikzlibrary{decorations.pathreplacing} 
\tikzset{
        edge/.style={thick, gray},
        medge/.style={decorate,very thick,decoration={snake}},
        nedge/.style={very thick,dashed,black},
        vertex/.style={shape=circle,thick,draw,node distance=3em}
}

\title{Fixed-Parameter Algorithms for Graph Constraint Logic} 
\author[1]{Tatsuhiko Hatanaka\thanks{Partially supported by JSPS KAKENHI Grant Number JP16J02175, Japan.}}
\author[2]{Felix Hommelsheim}
\author[1]{Takehiro Ito\thanks{Partially supported by JSPS KAKENHI Grant Numbers JP18H04091 and JP19K11814, Japan.}}
\author[3]{Yusuke Kobayashi\thanks{Partially supported by JSPS KAKENHI Grant Numbers 17K19960, 18H05291, and JP20K11692, Japan.}}
\author[4]{Moritz M\"uhlenthaler}
\author[1]{Akira Suzuki\thanks{Partially supported by JSPS KAKENHI Grant Numbers JP18H04091 and JP20K11666, Japan.}}

\affil[1]{Graduate School of Information Sciences, Tohoku University, Sendai, Japan}
\affil[2]{Fakult\"at f\"ur Mathematik, TU Dortmund University, Germany}
\affil[3]{Research Institute for Mathematical Sciences, Kyoto University, Japan}
\affil[4]{Laboratoire G-SCOP, Grenoble INP, Universit\'e Grenoble Alpes, France}

\begin{document}

\maketitle

\begin{abstract}
	Non-deterministic constraint logic (\ncl) is a simple model of
	computation based on orientations of a constraint graph with edge weights and
	vertex demands. \ncl captures \PSPACE\xspace and has been a useful tool for
	proving algorithmic hardness of many puzzles, games, and
	reconfiguration problems.  
	In particular, its usefulness stems from the fact that it remains \PSPACE-complete
	even under severe restrictions of the weights (e.g., only edge-weights one and two
	are needed) and the structure of the constraint graph (e.g., planar \andorsc
	graphs of bounded bandwidth). While such restrictions on the structure
	of constraint graphs do not seem to limit the expressiveness of \ncl,
	the building blocks of the constraint graphs cannot be limited without
	losing expressiveness: We consider as parameters the number of weight-one edges and
	the number of weight-two edges of a constraint graph, as well as the number of
	\andsc or \orsc vertices of an \andorsc constraint graph. We show that \ncl is
	fixed-parameter tractable (FPT) for any of these parameters. In particular, for \ncl
	parameterized by the number of weight-one edges or the number of \andsc
	vertices, we obtain a linear kernel. It follows that, in a sense, \ncl
	as introduced by Hearn and Demaine is defined in the most economical
	way for the purpose of capturing \PSPACE.
\end{abstract}

\newpage

\section{Introduction}
  Non-deterministic constraint logic (\ncl) has been introduced by Hearn and
  Demaine~\cite{HD:05} as a model of computation in order to show that many
  puzzles and games are complete in their natural complexity classes. For
  instance, they showed that the 1-player games Sokoban and Rush Hour are
  \PSPACE-complete~\cite{HD:05} and there are many follow-up results showing
  hardness of a large number of puzzles, games, and reconfiguration
  problems. An \ncl constraint graph is a graph with edge-weights
  one and two and a \emph{configuration} is given by an orientation of the
  constraint graph, such that the in-weight at each vertex is at least two. Two
  configurations are \emph{adjacent} if they differ with respect to the
  orientation of a single edge. The question whether two given configurations
  are connected by a path, i.e., a sequence of adjacent configurations, is
  known to be \PSPACE-complete, even if the constraint graph is a planar graph
  of maximum degree three (in fact, a planar \andorsc graph, to be defined
  shortly)~\cite{HD:05}. Similar hardness results are known for the
  question whether it is possible to reverse a single given edge, or whether
  there is a transformation between two configurations, such that each edge is
  reversed at most once.

  One of the main advantages of \ncl, apart from its simplicity, is its
  hardness on constraint graphs with a severely restricted structure, which
  entails strong hardness results for other problems. 
  In particular, \ncl is \PSPACE-complete on \emph{\andorsc graphs}, which are
  cubic graphs, where each vertex is either incident to three weight-two edges
  (``\orsc vertex'') or exactly one weight-two edge (``\andsc vertex''), see
  Figure~\ref{fig:andor}. It remains \PSPACE-complete if in addition we assume that the constraint graphs are
  planar~\cite{HD:05} and have bounded bandwidth~\cite{Zanden}.
  We investigate the possibility 
  of obtaining a further strengthening by
  restricting the \emph{composition} of the constraint graph. In particular we
  consider constraint graphs with a bounded number of weight-one or weight-two
  edges, and \andorsc graphs with a bounded number of AND or OR vertices.
  Our main result is that \ncl parameterized by any of the four quantities
  admits an FPT algorithm. That is, for the purpose of capturing \PSPACE, the definition of \ncl given by Hearn and Demaine is as
  economical as possible. We furthermore hope that based on our results, \ncl
  may become of interest for investigating the parameterized complexity of
  puzzles, games, and reconfiguration problems.
 
  In the following we adhere to the historical convention that an edge of
  weight one (resp., weight two) of a constraint graph is called \emph{red}
  (resp., \emph{blue}). We refer to the question whether a given configuration
  of a constraint graph is reachable from another given configuration as
  \emph{configuration-to-configuration} (\emph{\ctoc}). Furthermore, by
  \emph{configuration-to-edge} (\ctoe) we refer to the question whether, we can
  reach from a given configuration another one such that a given edge is
  reversed.

  \begin{figure}[tbh]
	  \begin{center}
		  \begin{subfigure}[b]{0.4\linewidth}
			  \centering
			  \begin{tikzpicture}[node distance=5em]
				\node[vertex,label=right:$\geq 2$] (c) {};  
				\node[above=3em of c] (u) {};  
				\node[below left=3em of c] (v) {};  
				\node[below right=3em of c] (w) {};  

				\draw[line width=2pt,red!80] (v) -- (c) node[midway,label=left:1] {};
				\draw[line width=2pt,red!80] (w) -- (c) node[midway,label=right:1] {};
				\draw[blue!80,line width=4pt] (u) -- (c) node[midway,label=right:2] {};
			  \end{tikzpicture}
			  \caption{\andsc vertex}
		  \end{subfigure}
		  \begin{subfigure}[b]{0.4\linewidth}
			  \centering
			  \begin{tikzpicture}[node distance=5em]
				\node[vertex,label=right:$\geq 2$] (c) {};  
				\node[above=3em of c] (u) {};  
				\node[below left=3em of c] (v) {};  
				\node[below right=3em of c] (w) {};  

				\draw[blue!80,line width=4pt] (v) -- (c) node[midway,label=left:2] {};
				\draw[blue!80,line width=4pt] (w) -- (c) node[midway,label=right:2] {};
				\draw[blue!80,line width=4pt] (u) -- (c) node[midway,label=right:2] {};
			  \end{tikzpicture}
			  \caption{\orsc vertex}
		  \end{subfigure}
		  \caption{The two types of vertices that occur in \andorsc constraint graphs. Edges must be oriented such that the in-weight at each vertex is at least two. By convention, weight-one edges are red and weight-two edges are blue.\label{fig:andor}}
	  \end{center}
  \end{figure}
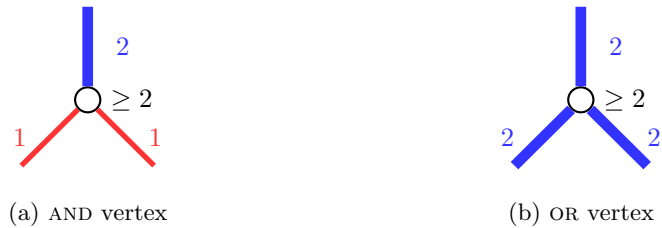


\paragraph*{Our Contribution}

\begin{table}[t]
	\vspace{-1.5em}
	\centering \caption{Parameterized Complexity of \ncl. For entries marked with $\kernelmark$ we obtain a linear kernel.\label{tab:FPT:overview}}
	\begin{tabular}{ l | l | l  }
		\toprule
		Parameter(s) & \ctoc & \ctoe \tabularnewline
		\midrule
		treewidth and max.~degree~\cite{Zanden} & \PSPACE-c  & \PSPACE-c \tabularnewline
		transformation length~\cite{Zanden}  & $\W[1]$-hard & $\W[1]$-hard \tabularnewline
		transformation length and max.~degree~\cite{Zanden} & \FPT & \FPT \tabularnewline
		\# of \andsc vertices (\andorsc graphs) & $\FPT^\kernelmark$~(Cor.~\ref{cor:ncl-and-or:red-edges-kernel}) & $\FPT^\kernelmark$~(Cor.~\ref{cor:red-edges-kernel:c2e}) \tabularnewline
		\# of \orsc vertices (\andorsc graphs) & \FPT~(Thm.~\ref{the:FPT_OR})  & \FPT~(Thm.~\ref{the:FPT_OR}) \tabularnewline
		\# of red edges & $\FPT^\kernelmark$~(Thm.~\ref{thm:ncl-red-edges-bound-kernel}) & $\FPT^\kernelmark$~(Cor.~\ref{cor:red-edges-kernel:c2e}) \tabularnewline
		\# of blue edges & \FPT~(Thm.~\ref{thm:blue01}) & \FPT~(Cor.~\ref{cor:blue01}) \tabularnewline
		\bottomrule
	\end{tabular}
\end{table}
We consider four natural parameterizations of \ncl and show that the
corresponding parameterized problems admit FPT algorithms. In particular we consider as parameters
\begin{enumerate}
	\item the number of \andsc vertices of an \andorsc graph,
	\item the number of \orsc vertices of an \andorsc graph,
	\item the number of red edges of a constraint graph, and
	\item the number of blue edges of a constraint graph.
\end{enumerate}
Note that none of these parameterizations trivially leads to an \XP~algorithm
that just enumerates all orientations for the constant number of red/blue edges
according to the parameter.
For an overview of the parameterized complexity results on \ncl, including our results,
please refer to Table~\ref{tab:FPT:overview}. 

\ncl is known to be \PSPACE-complete on \andorsc constraint graphs, which are
undirected edge-weighted graphs where each vertex is either and \andsc vertex
or an \orsc vertex as shown in Figure~\ref{fig:andor}. We show that \ctoc and \ctoe
parameterized by the number of \andsc vertices or the number of \orsc vertices
admits an FPT algorithm.
The algorithm first performs a preprocessing step followed by a reduction
to the problem \textsc{Binary Constraint Satisfiability Reconfiguration}
(\BCSR{} for short). Hatanaka et al.~have shown that \BCSR{} can be solved in
time $O^*(\doms^{O(\NB)})$, where $\doms$ and $\NB$ are the maximum size of a
domain and the number of non-Boolean variables, respectively~\cite{HIZ18}.

On general constraint graphs we obtain a linear kernel for \ctoc parameterized
by the number of red edges. For this purpose we introduce three reduction
rules, which, when applied exhaustively, yield a kernel of linear size. To the
best of our knowledge, this is the first polynomial kernel for a
parameterization of \ncl. 
The first rule states that each component containing at least two blue cycles
can be replaced by a gadget of constant size for each red edge that is attached
to the component. 
The second rule states that vertices incident to a blue edge only can be
deleted, since the orientation of this edge is the same for every orientation.
The third rule is inverse to subdividing a blue edge: any vertex incident to
precisely two blue edges can be deleted and replaced by a single blue edge
connecting its former neighbors.  Note that the number of red edges in an
\andorsc graph is precisely the number of \andsc vertices in an \andorsc graph.
Hence, a linear kernel for \ncl parameterized by the number of red edges
implies a linear kernel for \ncl parameterized by the number of \andsc vertices
of an \andorsc graph.  Furthermore, we show that slightly modified reduction
rules can be applied in order to obtain a linear kernel for \ctoe.

Finally, we consider \ctoc and \ctoe parameterized by the number $k$ of blue edges and
show that it admits an FPT algorithm.
Our key idea is to partition the set of feasible orientations of the constraint
graph into $2^{O(k)}$ classes, such that in each class, all blue edges are
oriented in the same way and the red edges have the same indegree sequence.
Denote the set of these classes by $\mathcal{F}$, and 
define a mapping $\phi$ from the set of orientations of the constraint graph to $\mathcal{F}$ (see Section~\ref{sec:classA} for the details). 
Then, in Section~\ref{sec:reconfF}, we define an adjacency relation 
between elements in $\mathcal{F}$, 
which is consistent with the reachability of configurations of the constraint graph in some sense.   
In our algorithm, instead of the original reconfiguration problem, 
we first solve the reconfiguration problem in $\mathcal{F}$, which can be done in $2^{O(k)} \cdot {\rm poly}(|V|)$ time, where $V$ is the set of vertices of the constraint graph.  
If it is impossible to reach the target configuration in $\mathcal{F}$, then
we can conclude that it is also impossible with respect to the original
constraint graph.  Otherwise, we can reduce the original problem to the case
that the blue edges agree in the initial and target configuration and the set
of red edges in the initial configurations whose orientation differs from the
target configuration consists of arc-disjoint dicycles (see
Section~\ref{sec:algo}). Finally, in Section~\ref{sec:blueprobA}, we test
whether the direction of each dicycle can be reversed or not. 

\vspace{-0.5em}
\paragraph*{Related Work}

A large number of puzzles, games, and
reconfiguration problems have been shown to be hard using reductions from \ncl
and its variants. Examples include motion planning problems, where rectangular
pieces have to be moved to certain final positions and sliding block puzzles
such as Rush Hour~\cite{FlakeBlaum,HD:05}, Sokoban~\cite{HD:05}, Snowman~\cite{HZY:17} and other
puzzle games such as Bloxors~\cite{vdZB:15}. In the \emph{bounded length}
version of \ncl, the orientation of each edge may be reversed at most once.
This variant has been used to show \NP-completeness of the games Klondike,
Mahjong Solitaire and Nonogram~\cite{Acyclic:14}. Note that \ncl gives a
uniform view on games as computation and often allows for simpler proofs and
strengthenings of known complexity results in this area.
Furthermore, 
deciding 
proof equivalence  in multiplicative linear logic has been shown to be 
\PSPACE-complete by a reduction from \ncl~\cite{HH:16}. 

\ncl is also very useful for showing hardness of reconfiguration problems. In
a reconfiguration problem we are given two configurations and agree on some
simple ``move'' that produces a new configuration from a given one. The
question is whether we can reach the second configuration from the first by a
sequence of moves.  For surveys on reconfiguration problems, please
refer to~\cite{Nishimura:18,Heuvel:13}. For many reconfiguration problems, such
as 
token sliding on graphs~\cite{HD:05}, a variant of 
independent set reconfiguration~\cite{Kaminski:12}, as well as 
vertex cover reconfiguration~\cite{Ito:11}, 
dominating set reconfiguration~\cite{Haddadan:16}, 
reconfiguration of paths~\cite{Paths:19},
and deciding Kempe-equivalence of 3-colorings~\cite{Kempe:19}, 
reductions from \ncl establish \PSPACE-hardness even on planar graphs of low
maximum degree.  Van der Zanden showed that there is some constant $c$, such
that \ncl is \PSPACE-complete on planar subcubic graphs of bandwidth at most
$c$~\cite{Zanden}. Note that this property is often maintained in the
reductions~\cite{Kempe:19,Paths:19,Haddadan:16,HD:05,Ito:11} 
and it implies 
that \ncl remains hard on graphs of bounded treewidth.

Tractable special cases of \ncl have received much less attention. 
Concerning parameterized complexity, \ncl remains \PSPACE-complete when
parameterized by treewidth and maximum degree of the constraint graph. On the
other hand, \ncl parameterized by the length of the transformation is
$\W[1]$-hard and it becomes FPT when parameterized by the length of the
transformation and the maximum degree~\cite{Zanden}.
If additionally each edge may be reversed at most once in a transformation,
\ncl is FPT when parameterized by treewidth and the maximum degree, or by the
length of the transformation~\cite{Zanden}.

\paragraph*{Organization}

The paper is organized as follows. In the next section we give some preliminaries
about \ncl and introduce notation used throughout the paper.
Section \ref{sec:OR} contains our FPT algorithm for \ncl parameterized by the number of \orsc vertices.
The linear kernel for \ncl parameterized by the number of red edges, which also
implies the result for \andsc vertices, can be found in Section~\ref{sec:red}.
Finally, in Section \ref{sec:blue} we give an FPT algorithm for \ncl parameterized by the number of blue edges.
Section~\ref{sec:conclusion} concludes the paper and gives some open problems.

\section{Preliminaries}

Let $G=(V, E)$ be an undirected graph, which may have multiple edges and (self)
loops.
We denote by $V(G)$ (resp., $E(G)$) the
set of vertices (resp., set of edges) $G$.  Each edge in an undirected
graph which joins two vertices $x$ and $y$ is represented as an unordered pair
$xy$ (or equivalently $yx$).  On the other hand, each arc in a digraph which
leaves $x$ and enters $y$ is written as an ordered pair $(x,y)$.  Let $(V, E, w)$ be a constraint graph, that is, an undirected graph $(V, E)$ with edge
weights $w : E \to \{1, 2\}$.  We denote by $\Ered$ and $\Eblue$ the
sets of red (weight one) and blue (weight two) edges in $E$, respectively, and have that $E = \Ered \cup \Eblue$.  We
denote by $\Vand(G)$ and $\Vor(G)$ the sets of \andsc and \orsc vertices in a
graph $G$, respectively; we sometimes drop $G$, and simply write $\Vand$ and
$\Vor$ if it is clear from the context.  A constraint graph is called
\emph{\andorsc graph} if each vertex is an \andsc or \orsc vertex; thus,
an \andorsc graph is $3$-regular.
 
An \emph{orientation} $A$ of $E$ is a multi-set of arcs obtained by replacing each edge in $E$ with a single arc having the same end vertices. 
We refer to $G$ as the underlying graph of the digraph $(V,A)$. 
For an orientation $A$ of $E$, we always denote by $\Ared$ and $\Ablue$ the subsets of $A$ corresponding to $\Ered$ and $\Eblue$, respectively. 
For any arc subset $B \subseteq A$ and a vertex $v \in V$, let $\indegree{B}(v)$ denote the number of arcs in $B$ that enter $v$. 
Then, $\indegree{B}$ can be regarded as a vector in $\mathbb{Z}_{\ge 0}^V$, where $\mathbb{Z}_{\ge 0}$ is the set of all nonnegative integers. 
An orientation $A$ of $E$ is {\em feasible} if $\indegree{\Ared}(v) + 2\cdot\indegree{\Ablue}(v) \ge 2$ for every $v \in V$; 
a feasible orientation is synonymously referred to as \emph{configuration}.

For two orientations $B$ and $B'$ of an edge subset $F \subseteq E$, we write $B \onestep B'$ if $B = B'$ or there exists an arc $(x,y) \in B$ such that $B' = (B \setminus \{(x,y)\}) \cup \{(y,x)\}$.  
For notational convenience, we simply write $B' = B - (x,y) + (y,x)$ in the latter case. 
For an orientation $B$ of $F$, \emph{reversing} the direction of an edge $xy
\in F$ is the operation which yields from $B$ an orientation $B'$ of $F$, such
that $B' = B-(x,y)+(y,x)$ if $(x, y) \in B$ and $B-(y,x)+(x,y)$ otherwise.
For two feasible orientations $A$ and $A^\prime$ of $E$, a sequence $\reconfseq{A_0, A_1, \ldots, A_\ell}$ of feasible orientations of $E$ is called a \emph{reconfiguration sequence} between $A$ and $A'$ if $A_0 = A$, $A_\ell = A'$, and $A_{i-1} \onestep A_i$ for all $i \in \{1,2,\ldots, \ell\}$.  
We write $A \reachable A'$ if there exists a reconfiguration sequence between
$A$ and $A'$ (or $A \not\leftrightsquigarrow  A'$ if not).  Given a constraint graph $G$ and two feasible orientations
$A_{\ini}$ and $A_{\tar}$ of $E(G)$, the problem \ctoc asks whether $A_{\ini}
\reachable A_{\tar}$ or not. Similarly, given a constraint graph $(G, w)$, a
feasible orientation $A_{\ini}$ of $E(G)$, and an edge $e \in E(G)$, the problem
\ctoe asks whether there is a feasible orientation $A_{tar}$, such that $A_{\ini}
\reachable A_{\tar}$ and the direction of $e$ is different in $A_{\ini}$ and
$A_{\tar}$.  
We denote by a triple $(G, A_{\ini}, A_{\tar})$ an instance of \ctoc and by a
triple $(G, A_{\ini}, vw)$ an instance of \ctoe.

\section{NCL for AND/OR graphs}
\label{sec:OR}
	In this section, we consider \ncl when restricted to \andorsc constraint graphs.
	Recall that \ncl remains \PSPACE-complete on \andorsc
	graphs~\cite{HD:05}.  We thus prove that \ctoc and \ctoe on \andorsc
	constraint graphs is fixed-parameter tractable when parameterized by
	the number of \orsc vertices. An analogous result for \ctoc and \ctoe
	parameterized by the number of \andsc vertices follows from our FPT
	result for \ncl parameterized by the number of red edges in the next
	section (see Theorem~\ref{thm:ncl-red-edges-bound-kernel}).  Therefore, the main result here is
	the following theorem.
	\begin{theorem} \label{the:FPT_OR}
		\ctoc  and \ctoe on \andorsc constraint graphs with $n$ vertices
		parameterized by the number $k$ of \orsc vertices admits a
		$2^{O(k)} \cdot {\rm poly}(n)$-time algorithm.
	\end{theorem}

	In the reminder of this section, we give an overview of the proof
	of Theorem~\ref{the:FPT_OR}.  Our strategy is to give an FPT-reduction
	from \ctoc on \andorsc constraint graphs to the \textsc{binary
	constraint satisfiability reconfiguration} problem (\BCSR, for
	short)~\cite{HIZ18}, which will be defined in
	Section~\ref{sec:reduction_BCSR}. To do so, we first apply some
	preprocessing to a given instance of \ctoc on an \andorsc graph (in
	Section~\ref{sec:modification}), and then give our FPT-reduction to
	\BCSR (in Section~\ref{sec:reduction_BCSR}). By similar arguments we obtain the result for \ctoe.
		
\subsection{Preprocessing}
\label{sec:modification}
	The preprocessing subdivides each blue edge that is not a loop into two
	blue edges. It is not hard to see that a single subdivision yields an
	equivalent instance: Let $uv$ be a blue edge of a constraint graph
	$\oriG$ and consider the constraint graph $G$ obtained by subdividing
	$uv$ into two blue edges $uz$ and $zv$, where $z$ is a new vertex we
	call \emph{middle vertex}.  Let $G$ be the resulting constraint graph
	and observe that from any feasible orientation $\oriA$ of $\oriG$ we
	may obtain a feasible orientation $A$ of $G$ by letting $A = \oriA -
	(u, v) + (u, z) + (z, v)$ if $(u, v) \in \oriA$ and $A = \oriA - (v, u)
	+ (v, z) + (z, u)$ otherwise.

	Furthermore, in any feasible orientation of $\oriG$, we can
	transfer in-weight from, say, $u$ to $v$ by reversing the arc $(v, u)$
	iff the in-weight at $u$ is at least four.
	Furthermore, due to the orientation of $uv$, the corresponding arc
	contributes to the in-weight of precisely one of $u$ and $v$.
	Conversely, in an orientation of $G$, we can transfer in-weight from,
	say, $u$ to $v$ by reversing the directions of the arcs corresponding
	to $uz$ and $zv$ iff the in-weight at $u$ is at least four.
	Furthermore, in any orientation of $G$, the arcs corresponding to $uz$
	and $zv$ contribute in-weight to at most one of $u$ and $v$. Hence, by
	subdividing a blue edge of $\oriG$ from an instance $(\oriG,
	\oriA_\ini, \oriA_\tar)$ of \ctoc, we obtain an equivalent instance.
	Let $(G, A_\ini, A_\tar)$ be the instance of \ctoc obtained from
	$(\oriG, \oriA_\ini, \oriA_\tar)$ by subdividing each blue edge of
	$\oriG$ that is not a loop. By repetition of the above argument we
	obtain the following result.

	\begin{lemma} \label{clm:cons2}
		$(G, A_{\ini}, A_{\tar})$ is a $\yes$-instance if and only if $(\oriG, \oriA_{\ini}, \oriA_{\tar})$ is.
	\end{lemma}

\subsection{FPT-reduction to \BCSR}
\label{sec:reduction_BCSR}
	In this subsection, we sketch our FPT-reduction to \BCSR.  We start by
	formally defining the problem \BCSR.  Let $H=(\Var,\Ecns)$ be an
	undirected graph.  We call each vertex $x \in \Var$ a \emph{variable}.
	Each $x \in \Var$ has a finite set $\Dom{x}$, called a \emph{domain of}
	$x$.  A variable $x$ is called a \emph{Boolean variable} if $|\Dom{x}|
	\le 2$, and  otherwise called a \emph{non-Boolean variable}.  Each edge
	$xy \in \Ecns$ has a subset $\Cns{xy} \subseteq \Dom{x}\times \Dom{y}$,
	called a \emph{\textup{(}binary\textup{)} constraint of} $xy$.  A
	mapping $\mapf \colon \Var \to \bigcup_{x \in \Var} \Dom{x}$ is a
	\emph{solution} of $H$ if $\mapf(x) \in \Dom{x}$ for every $x \in
	\Var$.  In addition, a solution $\mapf$ of $H$ is \emph{proper} if
	$\mapf(x)\mapf(y) \in \Cns{xy}$ for every $xy \in \Ecns$.  For two
	solutions $\mapf$ and $\mapf'$, we write $\mapf \onestep \mapf'$ if
	$|\{x \in \Var \midcolon \mapf(x) \ne \mapf'(x)\}|=1$.	Given an
	undirected graph $H$, a domain $\Dom{x}$ for each $x \in \Var$, a
	constraint $\Cns{xy}$ for each $xy \in \Ecns$, and two proper solutions
	$\mapf_{\ini}$ and $\mapf_{\tar}$ of $H$, the \textsc{binary constraint
	satisfiability reconfiguration} problem (\BCSR) asks whether there
	exists a sequence $\reconfseq{\mapf_0,\mapf_1,\ldots,\mapf_\ell}$ of
	proper solutions of $H$ such that $\mapf_0=\mapf_\ini$,
	$\mapf_\ell=\mapf_\tar$, and $\mapf_{i-1} \onestep \mapf_i$ for each $i
	\in \{1,2,\ldots, \ell\}$.  Let 
	$(H,\Dommap,\Cnsmap,\mapf_\ini,\mapf_\tar)$ an instance of \BCSR. 

	It is known that \BCSR can be solved in time $O^*(\doms^{O(\NB)})$,
	where $\doms := \max_{x \in X} |\Dom{x}|$ and $\NB$ is the number of
	non-Boolean variables in $X$~\cite[Theorem~18]{HIZ18}.  To prove
	Theorem~\ref{the:FPT_OR}, given an instance $(\oriG, \oriA_\ini,
	\oriA_\tar)$ of \ctoc on an \andorsc constraint graph with at most $k$
	\andorsc vertices, we first perform the preprocessing from
	Section~\ref{sec:modification} to obtain an instance $(G, A_\ini,
	A_\tar)$ of \ctoc. Note that $G$ is not an \andorsc graph, and $V$ can
	be partitioned into $\Vand(G)$, $\Vor(G)$ and $\Vmid(G)$, where
	$\Vmid(G)$ (or simply $\Vmid$) is the set of middle vertices in $G$.
	By Lemma~\ref{clm:cons2}, we have that $(G, A_\ini, A_\tar)$ is a
	\yes-instance if and only if $(\oriG, \oriA_{\ini}, \oriA_{\tar})$. 
	Hence, to conclude the proof of Theorem~\ref{the:FPT_OR}, we provide an
	FPT-reduction from a preprocessed instance $(G, A_{\ini}, A_{\tar})$ of
	\ctoc with the parameter $|\Vor(G)| = |\Vor(\oriG)| \le k$ to an
	instance $(H,\Dommap,\Cnsmap,\mapf_\ini,\mapf_\tar)$ of \BCSR such that
	both $\doms$ and $\NB$ are bounded by some computable functions
	depending only on $k$.  

	Due to the preprocessing,  observe that the constraint
	graph $G$ has no two parallel blue edges.  In addition, no
	edge in $G$ joins an \andsc vertex and an \orsc vertex,
	and hence we can partition $E$ into two sets $\Eand$ and $\Eor$,
	defined as follows: $\Eand$ is the set of edges of $G$ that are
	incident to an \andsc vertex; $\Eor$ is the set of edges
	of $G$ that are incident to an \orsc vertex. The high-level idea
	of the reduction to \BCSR{} is the following. For each \orsc vertex
	$v$, we create an \orsc variable $x_v$. Observe that the in-weight
	requirement at $v$ is violated only if each arc is pointing away from
	$v$. We forbid such orientations by giving each \orsc variable $x_v$ a
	domain of size seven corresponding to the seven legal orientations of
	the incident edges of $v$. 
	
	The remaining in-weight requirements and consistency requirements are
	modelled by adding constraints, which also define the set of edges in
	$H$. For each edge $e$ of $G$, we create a Boolean edge-variable $x_e$,
	whose domain represents the two possible orientations of an edge.  The
	construction of domains above ensures that in-weight requirement is
	satisfied for each \andsc vertex. To ensure the same property for all
	other vertices, we add three types of constraints for middle vertices
	and \andsc vertices, to enforce the following constraints:
	\begin{description}
		\item[Type 1:] Constraints for middle vertices.
		
		Let $v$ be a middle vertex between two vertices $v_1$ and $v_2$.
		Since both $v_1 v$ and $v v_2$ are blue edges, the in-weight requirement at $v$ is satisfied
		if and only if $v_1v$ or $vv_2$ points to $v$. 

		\item[Type 2-1:] Constraints for \andsc vertices having loops.
		
            Let $v$ be an \andsc vertex having a loop $vv$. So $vv$ must be red and the remaining edge $vv_3 \in \Eand$ is blue where $v_3$ is a middle vertex. 
            Then, the in-weight requirement at $v$  is satisfied if and only if $vv_3$ is oriented towards $v$.
            
		\item[Type 2-2:] Constraints for \andsc vertices without loops.

	    Let $v$ be an \andsc vertex, and let $vv_1$, $vv_2$, $vv_3$ be
	    three (distinct) edges incident to $v$ such that $vv_1$ and $vv_2$
	    are red, and $vv_3$ is blue; it may hold that $v_1 =v_2$.  Then,
	    the in-degree requirement at $v$ is satisfied if and only if i)
	    $vv_1$ or $vv_3$ are oriented towards $v$ and ii) $vv_2$ or $vv_3$
	    are oriented towards $v$. 
	\end{description}

	By the construction of constraints above, we know that a solution
	$\mapf$ of $H$ is proper if and only if the corresponding orientation
	$A_{\mapf}$ of $E$ is feasible.  Therefore, we can define proper
	solutions $\mapf_{\ini}$ and $\mapf_{\tar}$ of $H$ which correspond to
	feasible orientations $A_{\ini}$ and $A_{\tar}$ of $E$, respectively.
	In this way, from a preprocessed instance $(G, A_{\ini}, A_{\tar})$ of
	\ctoc with the parameter $|\Vor(G)| \le k$, we have constructed in
	polynomial time a corresponding equivalent instance
	$(H,\Dommap,\Cnsmap,\mapf_\ini,\mapf_\tar)$ of \BCSR such that $\doms =
	\max_{x \in X} |\Dom{x}| = 7$ and $\NB \le |\Vor(G)| \le k$.

\section{\ncl parameterized by the number of red edges}
\label{sec:red}
	Our main result in this section is a linear kernel for
	\ctoc parameterized by the number of red edges of the constraint graph.

	\begin{theorem} \label{thm:ncl-red-edges-bound-kernel}
		There is a polynomial-time algorithm that, given an instance
	of \ctoc on a constraint graph with $k$ red edges, outputs an equivalent instance of \ctoc of size $O(k)$.
	\end{theorem}

	In particular, Theorem~\ref{thm:ncl-red-edges-bound-kernel} implies
	that \ctoc parameterized by the number of red edges admits a
	$O^*(2^{O(k)})$-time algorithm. By observing that in any \andorsc
	constraint graph, the number of red edges is equal to the number of
	\andsc vertices, we immediately obtain the following result.
\begin{corollary} \label{cor:ncl-and-or:red-edges-kernel}
	\ctoc on \andorsc graphs parameterized by the number $k'$ of \andsc vertices admits a kernel of size $O(k')$.
\end{corollary}

\begin{figure}
    \centering
    \begin{tikzpicture}[vertex/.style={shape=circle,thick,draw,node distance=3em}]
        \node[vertex,label=right:$v_1$, fill=RoyalBlue] (v1) {};
        \node[vertex,label=above right:$v_0$,above right of=v1, fill= RoyalBlue] (v0) {};
        \node[vertex,label=below:$r$,right of=v0, fill=Red] (r) {};
        \node[vertex,label=left:$v_2$,left of=v1, fill=RoyalBlue] (v2) {};
        \node[vertex,label=right:$v_4$,above left of=v0, fill=RoyalBlue] (v4) {};
        \node[vertex,label=left:$v_3$, left of=v4, fill=RoyalBlue] (v3) {};
        \draw[->,line width=1mm,color=RoyalBlue] (v0) -- (r);
        \draw[->,line width=1mm,color=RoyalBlue] (v0) -- (v1);
        \draw[->,line width=1mm,color=RoyalBlue] (v1) -- (v2);
        \draw[<-,line width=1mm,color=RoyalBlue] (v1) -- (v3);
        \draw[<-,line width=1mm,color=RoyalBlue] (v0) -- (v4);        
        \draw[<-,line width=1mm,color=RoyalBlue] (v2) -- (v4);        
        \draw[->,line width=1mm,color=RoyalBlue] (v3) -- (v4);       
        \draw[->,line width=1mm,color=RoyalBlue] (v2) -- (v3);
    \end{tikzpicture}
    \caption{The gadget used in reduction rule \ref{red-rule:twocycles}.\label{fig:felix}}
\end{figure}
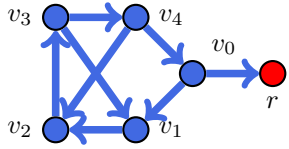

	It can be shown by similar arguments that there is also a linear kernel
	for \ctoe parameterized by the number of red edges of the constraint
	graph. In the remainder of this section, we prove
	Theorem~\ref{thm:ncl-red-edges-bound-kernel}.  Let $I = (G, A_{\ini},
	A_{\tar})$ be an instance of \ctoc, where $G$ is any constraint graph
	with $k$ red edges.  
	We give four reduction rules, and show that applying
	them repeatedly preserves the answer.  Furthermore, we show that
	applying them exhaustively yields an instance of size $O(k)$, where $k
	= |\Ered|$. To conclude the proof, we show that the reduction can be
	applied in polynomial time.

	We say that a vertex is \emph{blue} if all its incident edges are blue.
	Otherwise, if at least one incident edge is red, we call the
	vertex \emph{red}.  A subset $V' \subseteq V$ is called a \emph{blue
	component} if it is a connected component in the graph $(V, E^{ \rm
	blue})$. Note that a blue component may contain red vertices of $G$. 
	The first reduction rule removes blue components of $G$ that are directed cycles. Observe that no arc in such a component can be reversed. The second reduction rule removes blue components that
	contain at least two cycles and attaches to each red vertex $v$ of the
	component a copy of the gadget shown in Figure~\ref{fig:felix}.  
	The
	gadget consists of a cycle on five new vertices $\{v_0, v_1, v_2, v_3,
	v_4 \}$ with two chords $\{v_1, v_3 \}$ and $\{v_2, v_4 \}$.
	Additionally we add an edge joining $v$ and $v_0$. All edges of the
	gadget have weight two. The third reduction rule removes blue vertices
	of degree one and the last rule removes the center vertex of a blue
	path on three vertices.
	
	While modifying $G$ we also modify $A_{\rm ini}$
	and $A_{\rm tar}$ accordingly.  That is, if we delete edges of $G$,
	these edges are also deleted in $A_{\rm ini}$ and $A_{\rm tar}$.  If we
	add a gadget to $G$, then the arcs in $A_{\rm ini}$ and $A_{\rm tar}$ have
	the same orientation on the  gadget.
 Note that the number $k$ of red vertices is
not altered by an application of any of the rules.
Here is a more formal description of the four rules:

\begin{reduction}
	\label{red-rule:onecycle}
	Let $C$ be a component of $G$ that is a blue chordless cycle.  If the
	orientations $A_{\rm ini}$ and $A_{\rm tar}$ agree on $C$, then we
	remove $C$ from the graph and adjust $A_{\rm ini}$ and $A_{\rm tar}$
	accordingly. Otherwise we output a \NO-instance.	
\end{reduction}

\begin{reduction}
  \label{red-rule:twocycles}
  Let $C$ be a blue component that contains at least two 
  cycles. Then we remove from $G$ every blue vertex in $C$ and attach to each red
  vertex in $C$ a copy of the gadget in \figurename~\ref{fig:felix}. 
  Additionally we modify $A_{\rm ini}$ and $A_{\rm tar}$ accordingly such that both agree on each copy of the gadget.
\end{reduction}

\begin{reduction}
  \label{red-rule:degree1}
  If $G$ has a blue vertex $v$ of degree one, delete $v$ and its incident edge from $G$ and remove the corresponding arc(s) from $A_\ini$ and $A_\tar$.
\end{reduction}

\begin{reduction}
  \label{red-rule:2-path}
  Suppose $G$ has a blue vertex $v$ of degree 2, such that the two neighbors
  $u$ and $w$ of $v$ are non-adjacent in $G$. Then delete $v$ and its incident
  edges from $G$ and add the blue edge $uw$. Remove any arcs incident to $v$ from
  $A_{\ini}$ and $A_{\tar}$. Finally, add $(u, w)$ to $A_{\ini}$ (resp.
  $A_\tar$) if $(u, v) \in A_\ini$ (resp., $A_\tar$) and $(w, u)$ otherwise.
\end{reduction}

We show that applying any of the four rules is \emph{safe}, that is any
application results in a $\yes$-instance if and only if $I$ is a $\yes$-instance.

\begin{proposition}
  \label{prop:reduction}
  Reduction rules~\ref{red-rule:onecycle}--\ref{red-rule:2-path} are safe for \ctoc.
\end{proposition}

By applying a depth-first-search we can check if any of the rules can be applied. 
Thus we have the following.

\begin{proposition}
  \label{prop:running-time:rules} Reduction rules 
  \ref{red-rule:onecycle}-\ref{red-rule:2-path} can be applied exhaustively 
  in time $O(|V| \cdot (|V| + |E|))$.
\end{proposition}

Theorem~\ref{thm:ncl-red-edges-bound-kernel} now follows by the previous 
propositions and a simple counting argument.
Using similar arguments we show that there is also a linear kernel for \ctoe parameterized by the number $k$ of red edges. The main difference is that in 
reduction rule~\ref{red-rule:twocycles} we only add the gadget to each red vertex that is
part of a cycle or connected to two distinct cycles by two disjoint paths.
Furthermore, if the edge $e$ that we wish to reverse is part of of a component
containing two cycles, we add a gadget to the tail of $e$.

\begin{corollary}
	\label{cor:red-edges-kernel:c2e}
	\ctoe parameterized by the number $k$ of red edges admits a kernel of
	size $O(k)$. Furthermore, \ctoe on \andorsc graphs parameterized by the
	number $k'$ of  \andsc vertices admits a kernel of size $O(k')$.
\end{corollary}

\section{\ncl parameterized by the number of blue edges}
\label{sec:blue}

The objective of this section is to show that \ctoc parameterized by the number $k$ of
blue edges is fixed parameter tractable.

\begin{theorem}\label{thm:blue01}
	\ctoc parameterized by the number $k$ of blue edges can be solved in time
	$2^{O(k)} \cdot {\rm poly}(|V|)$.
\end{theorem}

In the remainder of this section, we prove Theorem~\ref{thm:blue01}. 
Let $I = (G, A_{\ini}, A_{\tar})$ be an instance of \ctoc, where $G$ is any constraint graph
and denote by $V$ and $E$ the set of vertices and edges of $G$, respectively.

Let $\mathcal{A}$ denote the set of all feasible orientations of $E$. 
Our key idea is to classify the feasible orientations into $2^{O(k)}$ classes, 
where each class is determined by the orientation $A^{\rm blue}$ of $E^{\rm blue}$ and 
the indegree sequence of $A^{\rm red}$.
Denote the set of these classes by $\mathcal{F}$, and 
define a mapping $\phi$ from $\mathcal{A}$ to $\mathcal{F}$ (see Section~\ref{sec:classA} for details). 
Then, in Section~\ref{sec:reconfF}, we define a reconfiguration relation $\underset{\mathcal{F}}{\leftrightsquigarrow}$ between elements in $\mathcal{F}$, 
which is consistent with $\leftrightsquigarrow$ in some sense.   
In our algorithm, instead of the original reconfiguration problem in $\mathcal{A}$, 
we first solve the reconfiguration problem in $\mathcal{F}$, which can be done in $2^{O(k)} \cdot {\rm poly}(|V|)$ time.  
If it has no reconfiguration sequence, then
we can conclude that there is no reconfiguration sequence in the original problem.  
Otherwise, we can reduce the original problem to the case when  $A^{\rm blue}_{\rm ini} = A^{\rm blue}_{\rm tar}$ and 
$A^{\rm red}_{\rm ini} \setminus A^{\rm red}_{\rm tar}$ consists of arc-disjoint dicycles (see Section~\ref{sec:algo}). 
Finally, in Section~\ref{sec:blueprobA}, we test whether the direction of  each dicycle can be reversed or not. 

Before starting the main part of the proof of Theorem~\ref{thm:blue01}, 
we show the following lemma that plays an important role in our argument. 
Roughly, it says that we can change the orientation of $E^{\rm red}$ keeping a certain indegree constraint. 

\begin{lemma}
\label{lem:blue02}
Let $A^{\rm red}_{\rm ini}$ and $A^{\rm red}_{\rm tar}$ be orientations of $E^{\rm red}$. 
Then, there exists a sequence $A^{\rm red}_0, A^{\rm red}_1, \dots , A^{\rm red}_l$ of orientations of $E^{\rm red}$
such that $A^{\rm red}_0 = A^{\rm red}_{\rm ini}$, $\rho_{A^{\rm red}_l} = \rho_{A^{\rm red}_{\rm tar}}$, $A^{\rm red}_{i-1} \leftrightarrow A^{\rm red}_i$ for $i=1, \dots , l$, and 
$\rho_{A^{\rm red}_i}(v) \ge \min \{ \rho_{A^{\rm red}_{\rm ini}}(v), \rho_{A^{\rm red}_{\rm tar}}(v)\}$ for any $v \in V$ and any $i \in \{0, 1, \dots , l\}$. 
\end{lemma}
\begin{proof}
We prove the lemma by induction on 
$|A^{\rm red}_{\rm ini} \setminus A^{\rm red}_{\rm tar}|$. 
If $\rho_{A^{\rm red}_{\rm ini}} = \rho_{A^{\rm red}_{\rm tar}}$, then the claim is obvious, because 
the sequence consisting of only one orientation $A^{\rm red}_0 = A^{\rm red}_{\rm ini}$ satisfies the conditions. 
Thus, it suffices to consider the case when $\rho_{A^{\rm red}_{\rm ini}} \not= \rho_{A^{\rm red}_{\rm tar}}$. 
In this case, there exists a vertex $u \in V$ such that $\rho_{A^{\rm red}_{\rm ini}} (u) > \rho_{A^{\rm red}_{\rm tar}}(u)$, 
because 
$\sum_{v \in V} \rho_{A^{\rm red}_{\rm ini}} (v) = \sum_{v \in V}  \rho_{A^{\rm red}_{\rm tar}}(v)$. 
Then, there exists an arc $a \in A^{\rm red}_{\rm ini} \setminus A^{\rm red}_{\rm tar}$ that enters $u$. 
Let $A^{\rm red}_1$ be the orientation of $E^{\rm red}$ obtained from $A^{\rm red}_{\rm ini}$ by reversing the direction of $a$. 
Since $|A^{\rm red}_1 \setminus A^{\rm red}_{\rm tar}|  < |A^{\rm red}_{\rm ini} \setminus A^{\rm red}_{\rm tar}|$, by induction hypothesis, 
there exists a sequence $A^{\rm red}_1, \dots , A^{\rm red}_l$ of orientations of $E^{\rm red}$
such that $\rho_{A^{\rm red}_l} = \rho_{A^{\rm red}_{\rm tar}}$, $A^{\rm red}_{i-1} \leftrightarrow A^{\rm red}_i$ for $i=2, \dots , l$, and 
$\rho_{A^{\rm red}_i}(v) \ge \min \{ \rho_{A^{\rm red}_1}(v), \rho_{A^{\rm red}_{\rm tar}}(v)\}$ for any $v \in V$ and any $i \in \{1, \dots , l\}$. 
By letting $A^{\rm red}_0 = A^{\rm red}_{\rm ini}$, the sequence $A^{\rm red}_0, A^{\rm red}_1, \dots , A^{\rm red}_l$ satisfies the conditions, 
because $A^{\rm red}_0 \leftrightarrow A^{\rm red}_1$, 
$\rho_{A^{\rm red}_1}(v) \ge \rho_{A^{\rm red}_{\rm ini}}(v)$ for each $v \in V \setminus \{u\}$, and 
$\min \{ \rho_{A^{\rm red}_1}(u), \rho_{A^{\rm red}_{\rm tar}}(u)\} = \rho_{A^{\rm red}_{\rm tar}}(u) = \min \{ \rho_{A^{\rm red}_{\rm ini}}(u), \rho_{A^{\rm red}_{\rm tar}}(u)\}$. 
\end{proof}

The proof of Lemma~\ref{lem:blue02} is constructive, and hence 
we can find such a sequence efficiently.

\subsection{Classification of $\mathcal A$}
\label{sec:classA}

In this subsection, we classify the feasible orientations into $2^{O(k)}$ classes. 
Let $X \subseteq V$ be the set of all vertices to which 
edges in $E^{\rm blue}$ are incident.  
Define $\mathcal F$ as the set of all pairs $(A^{\rm blue}, d)$ 
where $A^{\rm blue}$ is an orientation of $E^{\rm blue}$ and $d$ is a vector in $\{0, 1, 2\}^X$
satisfying the following conditions:  
\begin{enumerate}
\item[(1)]
$2 \rho_{A^{\rm blue}} (v) + d(v) \ge 2$ for any $v \in X$. 
\item[(2)]
There exists an orientation $A^{\rm red}$ of $E^{\rm red}$ such that for any $v \in V$,  
\[
\rho_{A^{\rm red}} (v)
\begin{cases}
= 0 & \mbox{if $v \in X$ and $d(v)=0$,} \\
= 1 & \mbox{if $v \in X$ and $d(v)=1$, and} \\
\ge 2 & \mbox{otherwise.}
\end{cases}
\]
\end{enumerate}
We note that $|\mathcal {F}| \le 2^{|E^{\rm blue}|} \cdot 3^{|X|} = 2^{O(k)}$, because $|X| \le 2 |E^{\rm blue}|$. 
For a vector $d \in \{0, 1, 2\}^X$, 
we say that an orientation $A^{\rm red}$ of $E^{\rm red}$ {\em realizes} $d$ if 
$A^{\rm red}$ satisfies the condition (2) above. 
We can easily see that 
if $(A^{\rm blue}, d) \in \mathcal F$ holds and $A^{\rm red}$ realizes $d$, then 
$A := A^{\rm blue} \cup A^{\rm red}$ is a feasible orientation of $E$. 
Conversely, if 
$A = A^{\rm blue} \cup A^{\rm red}$ is a feasible orientation of $E$ (i.e., $A \in \mathcal{A}$), then 
the vector $d \in \{0, 1, 2\}^X$ defined by $d(v) = \min \{\rho_{A^{\rm red}}(v) , 2\}$ for each $v \in X$ 
satisfies that 
$(A^{\rm blue}, d) \in \mathcal F$. 
This defines a mapping $\phi$ from $\mathcal A$ to $\mathcal F$. 

We can also see that the membership problem of $\mathcal{F}$ can be decided in polynomial time. 

\begin{lemma}
\label{lem:blue04}
For an orientation $A^{\rm blue}$ of $E^{\rm blue}$ and a vector $d \in \{0, 1, 2\}^X$, 
we can test whether $(A^{\rm blue}, d) \in \mathcal F$ or not in polynomial time. 
\end{lemma}
\begin{proof}
We can easily check the condition (1). 
To check the condition (2), 
we construct a digraph $\hat G=(\hat V, \hat A)$ and consider a network flow problem in it. 
Introduce a new vertex $w_e$ for each $e \in E^{\rm red}$ and two new vertices $s$ and $t$, 
and define $\hat V := V \cup \{w_e \mid e \in E^{\rm red}\} \cup \{s, t\}$. 
Define the arc set $\hat A:= \hat A_1 \cup \hat A_2 \cup \hat A_3$ by
\begin{align*}
\hat A_1 &:= \{(s, w_e) \mid e \in E^{\rm red} \}, \\
\hat A_2 &:= \{(w_e, v) \mid e \in E^{\rm red}, v\in V, \mbox{$e$ is incident to $v$ in $G$} \}, \\
\hat A_3 &:= \{(v, t) \mid v\in V \}. 
\end{align*}
For each $a \in \hat A$, define the lower bound $l(a)$ and the upper bound $u(a)$ of 
the amount of flow through $a$ as follows. 
\begin{itemize}
\item
For each $(s, w_e) \in \hat A_1$, define $l(s, w_e) := u (s, w_e) := 1$. 
\item
For each $(w_e, v) \in \hat A_2$, define $l(w_e, v) := 0$ and $u (w_e, v) := 1$. 
\item
For each $(v, t) \in \hat A_3$, define $l(v, t) := u (v, t) := d(v)$ if $v \in X$ and $d(v) \in \{0, 1\}$, 
and define $l(v, t) := 2$ and $u (v, t) := + \infty$ otherwise. 
\end{itemize}
Then, the condition (2) holds if and only if $\hat G$ has an integral $s$-$t$ flow satisfying the above constraint.  
This can be tested in polynomial time 
by a standard maximum flow algorithm (see e.g. \cite[Corollary 11.3a]{Schrijver:03}). 
\end{proof}

\subsection{Reconfiguration in $\mathcal{F}$}
\label{sec:reconfF}

In this subsection, we consider a reconfiguration between elements in $\mathcal{F}$. 
For $(A^{\rm blue}_1, d_1), (A^{\rm blue}_2, d_2) \in \mathcal F$, 
we denote $(A^{\rm blue}_1, d_1) \xleftrightarrow[\mathcal F]{} (A^{\rm blue}_2, d_2)$
if 
\begin{itemize}
\item
$d_1 = d_2$ and $A^{\rm blue}_1 \leftrightarrow A^{\rm blue}_2$, or 
\item
$A^{\rm blue}_1 = A^{\rm blue}_2$. 
\end{itemize}
If there exists a sequence 
$(A^{\rm blue}_0, d_0), (A^{\rm blue}_1, d_1), \dots , (A^{\rm blue}_l, d_l) \in \mathcal F$ such that 
$(A^{\rm blue}_{i-1}, d_{i-1}) \xleftrightarrow[\mathcal F]{} (A^{\rm blue}_i, d_i)$ for $i=1, \dots , l$, then
we denote 
$(A^{\rm blue}_0, d_0) \underset{\mathcal{F}}{\leftrightsquigarrow} (A^{\rm blue}_l, d_l)$. 
Then, we can easily see the following. 
\begin{lemma}\label{lem:blue08}
	Let $A_{\rm ini}, A_{\rm tar} \in \mathcal{A}$.
	If $A_{\rm ini} \leftrightsquigarrow A_{\rm tar}$, then 
	$\phi(A_{\rm ini}) \underset{\mathcal{F}}{\leftrightsquigarrow} \phi(A_{\rm tar})$. 
\end{lemma}
\begin{proof}
	If $A_{\rm ini} \leftrightarrow A_{\rm tar}$, then 
	$\phi(A_{\rm ini}) \xleftrightarrow[\mathcal F]{} \phi(A_{\rm tar})$ by definition. 
	By using this relationship repeatedly, we obtain the claim. 
\end{proof}

Although the opposite implication is not true, 
we show the following statement. 
\begin{lemma}
	\label{lem:blue05}
	Let $A_{\rm ini}, A_{\rm tar} \in \mathcal{A}$.  
	If $\phi(A_{\rm ini}) \underset{\mathcal{F}}{\leftrightsquigarrow} \phi(A_{\rm tar})$, then 
	there exists $A^\circ_{\rm tar} \in \mathcal{A}$ such that
	$\phi(A^\circ_{\rm tar}) = \phi(A_{\rm tar})$ and
	$A_{\rm ini} \leftrightsquigarrow A^\circ_{\rm tar}$. 
\end{lemma}
\begin{proof}
It suffices to consider the case when 
$\phi(A_{\rm ini}) \xleftrightarrow[\mathcal F]{} \phi(A_{\rm tar})$. 
Denote $\phi(A_{\rm ini}) = (A^{\rm blue}_{\rm ini}, d_{\rm ini})$ and $\phi(A_{\rm tar}) = (A^{\rm blue}_{\rm tar}, d_{\rm tar})$. 
By definition, we have either 
$d_{\rm ini} = d_{\rm tar}$ and $A^{\rm blue}_{\rm ini} \leftrightarrow A^{\rm blue}_{\rm tar}$, or 
$A^{\rm blue}_{\rm ini} = A^{\rm blue}_{\rm tar}$. 

If $d_{\rm ini} = d_{\rm tar}$ and $A^{\rm blue}_{\rm ini} \leftrightarrow A^{\rm blue}_{\rm tar}$, then 
$A_{\rm ini} \leftrightarrow A^{\rm blue}_{\rm tar} \cup A^{\rm red}_{\rm ini}$ and 
$\phi(A^{\rm blue}_{\rm tar} \cup A^{\rm red}_{\rm ini}) = (A^{\rm blue}_{\rm tar}, d_{\rm ini}) = (A^{\rm blue}_{\rm tar}, d_{\rm tar}) = \phi(A_{\rm tar})$, 
which means that $A^\circ_{\rm tar} := A^{\rm blue}_{\rm tar} \cup A^{\rm red}_{\rm ini}$ satisfies the conditions. 

Otherwise, let $A^{\rm blue} := A^{\rm blue}_{\rm ini} = A^{\rm blue}_{\rm tar}$. 
By Lemma~\ref{lem:blue02},  
we obtain a sequence $A^{\rm red}_0, A^{\rm red}_1, \dots , A^{\rm red}_l$ of orientations of $E^{\rm red}$
such that $A^{\rm red}_0 = A^{\rm red}_{\rm ini}$, $\rho_{A^{\rm red}_l} = \rho_{A^{\rm red}_{\rm tar}}$, $A^{\rm red}_{i-1} \leftrightarrow A^{\rm red}_i$ for $i=1, \dots , l$, and 
$\rho_{A^{\rm red}_i}(v) \ge \min \{ \rho_{A^{\rm red}_{\rm ini}}(v), \rho_{A^{\rm red}_{\rm tar}}(v)\}$ for any $v \in V$ and any $i \in \{0, 1, \dots , l\}$.
Then, for any $i \in \{0, 1, \dots , l\}$, we have 
\begin{align*}
2 \rho_{A^{\rm blue}} (v) + \rho_{A^{\rm red}_i}(v) 
\ge \min \{2 \rho_{A^{\rm blue}} (v) + \rho_{A^{\rm red}_{\rm ini}}(v), 2 \rho_{A^{\rm blue}} (v) + \rho_{A^{\rm red}_{\rm tar}}(v)\} 
\ge 2
\end{align*}
for any $v \in V$,    
and hence $A^{\rm blue} \cup A^{\rm red}_i$ is feasible. 
Since $A^{\rm blue} \cup A^{\rm red}_{i-1} \leftrightarrow A^{\rm blue} \cup A^{\rm red}_i$ for $i=1, \dots , l$, 
we have 
$$
(A_{\rm ini} =) A^{\rm blue} \cup A^{\rm red}_{\rm ini} \leftrightsquigarrow A^{\rm blue} \cup A^{\rm red}_l.
$$
Furthermore, 
since $\rho_{A^{\rm red}_l} = \rho_{A^{\rm red}_{\rm tar}}$, 
we have $\phi(A^{\rm blue} \cup A^{\rm red}_l) = \phi(A_{\rm tar})$. 
Therefore, $A^\circ_{\rm tar} := A^{\rm blue} \cup A^{\rm red}_l$ satisfies the conditions in the lemma. 
\end{proof}

Note that
we can construct $A^\circ_{\rm tar}$ and a reconfiguration sequence in Lemma~\ref{lem:blue05}
efficiently.

\subsection{Algorithm}
\label{sec:algo}

Let $I = (G, A_{\ini}, A_{\tar})$ be an instance of \ctoc.
We first compute $\phi(A_{\rm ini})$ and $\phi(A_{\rm tar})$, 
and test whether 
$\phi(A_{\rm ini}) \underset{\mathcal{F}}{\leftrightsquigarrow} \phi(A_{\rm tar})$ or not. 
If $\phi(A_{\rm ini}) \not\underset{\mathcal{F}}{\leftrightsquigarrow} \phi(A_{\rm tar})$, 
then we can immediately conclude that $A_{\rm ini}  \not\leftrightsquigarrow A_{\rm tar}$
by Lemma~\ref{lem:blue08}.  

Thus, in what follows, suppose that $\phi(A_{\rm ini}) \underset{\mathcal{F}}{\leftrightsquigarrow} \phi(A_{\rm tar})$. 
In this case, by applying Lemma~\ref{lem:blue05}, 
we can construct $A^\circ_{\rm tar} \in \mathcal{A}$ with $\phi(A^\circ_{\rm tar}) = \phi(A_{\rm tar})$ such that 
$A_{\rm ini} \leftrightsquigarrow A^\circ_{\rm tar}$. 
This shows that 
$A_{\rm ini} \leftrightsquigarrow A_{\rm tar}$ is equivalent to $A^\circ_{\rm tar} \leftrightsquigarrow A_{\rm tar}$, 
which means that we can regard $A^\circ_{\rm tar}$ as a new initial configuration instead of $A_{\rm ini}$. 
Thus, the problem is reduced to the case when $\phi(A_{\rm ini}) = \phi(A_{\rm tar})$. 
In particular, we have $A^{\rm blue}_{\rm ini} = A^{\rm blue}_{\rm tar}$.

Suppose that $A^{\rm blue}_{\rm ini} = A^{\rm blue}_{\rm tar} =: A^{\rm blue}$ and
$\rho_{A^{\rm red}_{\rm ini}} \not= \rho_{A^{\rm red}_{\rm tar}}$. 
Then, by applying Lemma~\ref{lem:blue02},  
we obtain an orientation $A^{\rm red}_l$ of $E^{\rm red}$ such that  
$\rho_{A^{\rm red}_l} = \rho_{A^{\rm red}_{\rm tar}}$ and $A_{\rm ini} \leftrightsquigarrow A^{\rm blue} \cup A^{\rm red}_l$. 
This shows that 
$A_{\rm ini} \leftrightsquigarrow A_{\rm tar}$ is equivalent to $A^{\rm blue} \cup A^{\rm red}_l\leftrightsquigarrow A_{\rm tar}$, 
which means that we can regard $A^{\rm blue} \cup A^{\rm red}_l$ as a new initial configuration instead of $A_{\rm ini}$. 
Thus, the problem is reduced to the case when $\rho_{A^{\rm red}_{\rm ini}} = \rho_{A^{\rm red}_{\rm tar}}$. 
If $A^{\rm red}_{\rm ini} = A^{\rm red}_{\rm tar}$, 
we conclude that $A_{\rm ini} \leftrightsquigarrow A_{\rm tar}$. 
Otherwise, since $\rho_{A^{\rm red}_{\rm ini}} = \rho_{A^{\rm red}_{\rm tar}}$, the set $A^{\rm red}_{\rm ini} \setminus A^{\rm red}_{\rm tar}$ 
can be decomposed into arc-disjoint cycles. 

Note that all of the above procedures can be executed in $2^{O(k)} \cdot {\rm poly} (|V|)$ time, since $|\mathcal{F}| = 2^{O(k)}$. 
In what follows, we give an algorithm for testing whether the direction of each cycle can be reversed or not. 
For this purpose, we show the following lemma.

\begin{lemma}
\label{lem:blue06}
Let $A_{\rm ini} \in \mathcal{A}$ and 
let $C$ be a dicycle with all the arcs in $A^{\rm red}_{\rm ini}$. 
Then, the followings are equivalent. 
\begin{enumerate}
\item[(i)]
$A_{\rm ini} \leftrightsquigarrow (A_{\rm ini} \setminus C) \cup \overline{C}$, where $\overline{C}$ is the reverse dicycle of $C$. 
\item[(ii)]
For any arc $a$ in $C$, 
there exists an orientation $A \in \mathcal A$ such that $A_{\rm ini} \leftrightsquigarrow A$ and $a \not\in A$. 
\item[(iii)]
For any $u \in V(C)$, there exists an orientation $A \in \mathcal{A}$ 
such that 
$A_{\rm ini} \leftrightsquigarrow A$
and
$2 \rho_{A^{\rm blue}}(u) + \rho_{A^{\rm red}}(u) \ge 3$.   
\end{enumerate}
\end{lemma}
\begin{proof}
We prove (i)$\Rightarrow$(ii), (ii)$\Rightarrow$(iii), and (iii)$\Rightarrow$(i), respectively. 

{\bf [(i)$\Rightarrow$(ii)]} 
If (i) holds, then $A := (A_{\rm ini} \setminus C) \cup \overline{C}$ satisfies the conditions in (ii), since it contains no arc in $C$. 

{\bf [(ii)$\Rightarrow$(iii)]} 
We prove the contraposition. 
Assume that (iii) does not hold, that is, 
there exists a vertex $u \in V(C)$ such that 
$2 \rho_{A^{\rm blue}}(u) + \rho_{A^{\rm red}}(u) = 2$
for any $A \in \mathcal A$ with $A_{\rm ini} \leftrightsquigarrow A$. 
Let $a$ be the arc in $C$ that enters $u$. 
Since we cannot reverse the direction of $a$ without violating the feasibility, 
$a$ is contained in any orientation $A \in \mathcal A$ with $A_{\rm ini} \leftrightsquigarrow A$.

{\bf [(iii)$\Rightarrow$(i)]} 
Suppose that (iii) holds. 
We take a sequence 
$A_0, A_1, \dots , A_l$ of feasible orientations of $E$ such that 
$A_0 = A_{\rm ini}$, 
$A_i$ is obtained from $A_{i-1}$ by reversing an arc $a_i \in A_{i-1}$ for $i \in \{1, 2, \dots , l\}$, and 
there exists $u \in V(C)$ such that 
$2 \rho_{A^{\rm blue}_l}(u) + \rho_{A^{\rm red}_l}(u) \ge 3$.  
By taking a minimal sequence with these conditions, we may assume that 
$a_i$ is not contained in $C$ for $i \in \{1, 2, \dots , l\}$. 
Since 
$2 \rho_{A^{\rm blue}_l}(u) + \rho_{A^{\rm red}_l}(u) \ge 3$, 
starting from $A_l$, 
we can change the direction of each arc in $C$ one by one without violating the feasibility, 
which shows that $A_l \leftrightsquigarrow (A_l \setminus C) \cup \overline{C}$. 
On the other hand, since $(A_i \setminus C) \cup \overline{C}$ is obtained from $(A_{i-1} \setminus C) \cup \overline{C}$ 
by reversing $a_i$ for $i \in \{1, 2, \dots , l\}$, we obtain
$(A_{\rm ini} \setminus C) \cup \overline{C} \leftrightsquigarrow (A_l \setminus C) \cup \overline{C}$. 
Thus, it holds that $A_{\rm ini} \leftrightsquigarrow A_l \leftrightsquigarrow (A_l \setminus C) \cup \overline{C} \leftrightsquigarrow (A_{\rm ini} \setminus C) \cup \overline{C}$. 
\end{proof}

Let $C$ be a dicycle in $A^{\rm red}_{\rm ini} \setminus A^{\rm red}_{\rm tar}$. 
Fix a vertex $u \in V(C)$ and consider the following problem, 
for which an algorithm is presented later in Section~\ref{sec:blueprobA}.  

\begin{description}
\item[\ProblemA]
\item[Input:]
A constraint graph $G$, an orientation $A_{\rm ini} \in \mathcal A$, and a vertex $u \in V(G)$. 

\item[Task:]
Find an orientation $A \in \mathcal{A}$ 
s.t.\
$2 \rho_{A^{\rm blue}}(u) + \rho_{A^{\rm red}}(u) \ge 3$ 
and $A_{\rm ini} \leftrightsquigarrow A$ (if exists). 
\end{description}

If \ProblemA has no solution, 
then condition (iii) in Lemma~\ref{lem:blue06} does not hold. 
This shows that the condition (ii) in Lemma~\ref{lem:blue06} does not hold, that is, 
there exists an arc $a$ in $C$ 
that is contained in 
any orientation $A \in \mathcal A$ with $A_{\rm ini} \leftrightsquigarrow A$.  
In this case, since $a \in A_{\rm ini} \setminus A_{\rm tar}$, 
we conclude that 
$A_{\rm ini} \not\leftrightsquigarrow A_{\rm tar}$. 

Otherwise, \ProblemA has a solution, and hence 
the condition (iii) in Lemma~\ref{lem:blue06} holds. 
Since it is equivalent to the condition (i) in Lemma~\ref{lem:blue06}, we have that  
$A_{\rm ini} \leftrightsquigarrow (A_{\rm ini} \setminus C) \cup \overline{C}$. 
Therefore, 
$A_{\rm ini} \leftrightsquigarrow A_{\rm tar}$ is equivalent to $(A_{\rm ini} \setminus C) \cup \overline{C} \leftrightsquigarrow A_{\rm tar}$, 
which means that we can regard $(A_{\rm ini} \setminus C) \cup \overline{C}$ as a new initial configuration instead of $A_{\rm ini}$. 
Then, the problem is reduced to the case with smaller $|A_{\rm ini} \setminus A_{\rm tar}|$. 
By applying this procedure at most $O(|E|)$ times repeatedly, we can solve the original reconfiguration problem. 
The entire algorithm is shown in Algorithm~\ref{alg1}. 

\begin{algorithm}[t]
	\SetKwInOut{Input}{Input}\SetKwInOut{Output}{Output}
	  \AlgoDontDisplayBlockMarkers\SetAlgoNoLine\SetAlgoNoEnd
	\Input{a graph $G=(V,E)$ and orientations $A_{\rm ini}, A_{\rm tar} \in \mathcal A$.}
	\Output{ ``\YES'' if $A_{\rm ini} \leftrightsquigarrow A_{\rm tar}$, and ``\NO'' otherwise. }
	Compute $\mathcal{F}, \phi(A_{\rm ini})$, and $\phi(A_{\rm tar})$ ; \\
	\lIf{$\phi(A_{\rm ini}) \not\underset{\mathcal{F}}{\leftrightsquigarrow} \phi(A_{\rm tar})$}{ 
		\Return ``\NO''  
	}
	\If{$\phi(A_{\rm ini}) \not= \phi(A_{\rm tar})$ or $\rho_{A^{\rm red}_{\rm ini}} \not= \rho_{A^{\rm red}_{\rm tar}}$}{ 
		Compute $A \in \mathcal A$ such that $\phi(A) = \phi(A_{\rm tar})$,
         $\rho_{A^{\rm red}} = \rho_{A^{\rm red}_{\rm tar}}$, and $A_{\rm ini} \leftrightsquigarrow A$ \;
		$A_{\rm ini} \leftarrow A$	  
		\tcp*[r]{See Section~\ref{sec:reconfF}}
	}
	\While{$A^{\rm red}_{\rm ini} \setminus A^{\rm red}_{\rm tar}$ contains a dicycle $C$}
	{
		Take $u \in V(C)$ and solve \ProblemA \;
		\eIf{\ProblemA has no feasible solution}
		{ 
			Return ``\NO'' \;
		}
		{ 
			$A_{\rm ini} \leftarrow  (A_{\rm ini} \setminus C) \cup \overline{C}$ ;  \tcp*[f]{See Section~\ref{sec:algo}}
		}
	}
	\Return ``\YES'' \;
	\caption{Algorithm for the Reconfiguration Problem}
	\label{alg1}
\end{algorithm}

\subsection{Algorithm for \ProblemA}
\label{sec:blueprobA}

The remaining task is to 
give a polynomial time algorithm for \ProblemA. 
For this purpose, we use a similar argument to Section~\ref{sec:reconfF}. 
Suppose we are given a graph $G=(V, E)$ and a vertex $u \in V$. 
Recall that $X \subseteq V$ is the set of all vertices to which 
edges in $E^{\rm blue}$ are incident.  
Define $\mathcal{F}_u$ as the set of all pairs $(A^{\rm blue}, d)$, 
where $A^{\rm blue}$ is an orientation of $E^{\rm blue}$ and $d$ is a vector in $\{0, 1, 2, 3\}^{X \cup \{u\}}$ 
satisfying the following conditions:  
\begin{enumerate}
\item[(1)]
$2 \rho_{A^{\rm blue}} (v) + d(v) \ge 2$ for any $v \in X \cup \{u\}$. 
\item[(2)]
There exists an orientation $A^{\rm red}$ of $E^{\rm red}$ such that for any $v \in V$, 
\[
\rho_{A^{\rm red}} (v)
\begin{cases}
= d(v) & \mbox{if $v \in X \cup \{u\}$ and $d(v) \in \{0, 1, 2\}$,} \\
\ge 3 & \mbox{if $v \in X \cup \{u\}$ and $d(v)=3$, and} \\
\ge 2 & \mbox{if $v \in V \setminus (X \cup \{u\})$}.
\end{cases}
\]
\end{enumerate}
We note that $|\mathcal {F}_u| \le 2^{|E^{\rm blue}|} \cdot 4^{|X \cup \{u\}|} = 2^{O(k)}$. 
We define $\xleftrightarrow[\mathcal{F}_u]{}$, $\underset{\mathcal{F}_u}{\leftrightsquigarrow}$, and $\phi_u$ in the same way as 
$\xleftrightarrow[\mathcal{F}]{}$, $\underset{\mathcal{F}}{\leftrightsquigarrow}$, and $\phi$. 
We obtain the following lemmas in the same way as Lemmas~\ref{lem:blue04},~\ref{lem:blue08}, and~\ref{lem:blue05}. 

\begin{lemma}\label{lem:blue14}
For an orientation $A^{\rm blue}$ of $E^{\rm blue}$ and a vector $d \in \{0, 1, 2, 3\}^X$, 
we can test whether $(A^{\rm blue}, d) \in \mathcal{F}_u$ or not in polynomial time. 
\end{lemma}

\begin{lemma}\label{lem:blue18}
Let $A_{\rm ini}, A_{\rm tar} \in \mathcal{A}$.
If $A_{\rm ini} \leftrightsquigarrow A_{\rm tar}$, then 
$\phi_u(A_{\rm ini}) \underset{\mathcal{F}_u}{\leftrightsquigarrow} \phi_u(A_{\rm tar})$. 
\end{lemma}

\begin{lemma}\label{lem:blue15}
Let $A_{\rm ini}, A_{\rm tar} \in \mathcal{A}$.  
If $\phi_u(A_{\rm ini}) \underset{\mathcal{F}_u}{\leftrightsquigarrow} \phi_u(A_{\rm tar})$, then 
there exists $A^\circ_{\rm tar} \in \mathcal{A}$ with $\phi_u(A^\circ_{\rm tar}) = \phi_u(A_{\rm tar})$ such that 
$A_{\rm ini} \leftrightsquigarrow A^\circ_{\rm tar}$. 
\end{lemma}

\begin{proposition}
\label{prop:blue:A}
\ProblemA has a solution if and only if
there exists a pair $(A^{\rm blue}, d) \in \mathcal{F}_u$
such that $2 \rho_{A^{\rm blue}} (u) + d(u) \ge 3$ and $\phi_u(A_{\rm ini}) \underset{\mathcal{F}_u}{\leftrightsquigarrow} (A^{\rm blue}, d)$.
\end{proposition}
\begin{proof}
If $A$ is a solution of \ProblemA, then 
$\phi_u(A) = (A^{\rm blue}, d)$ satisfies the conditions by Lemma~\ref{lem:blue18}. 
Conversely, assume that there exists a pair $(A^{\rm blue}, d)\in \mathcal{F}_u$ 
such that $2 \rho_{A^{\rm blue}} (u) + d(u) \ge 3$ and $\phi_u(A_{\rm ini}) \underset{\mathcal{F}_u}{\leftrightsquigarrow} (A^{\rm blue}, d)$.  
By Lemma~\ref{lem:blue15},  
there exists an orientation $A \in \mathcal{A}$ with $\phi_u(A) = (A^{\rm blue}, d)$ such that 
$A_{\rm ini} \leftrightsquigarrow A$. 
Since $2 \rho_{A^{\rm blue}} (u) + \rho_{A^{\rm red}} (u) \ge 2 \rho_{A^{\rm blue}} (u) + d (u) \ge 3$, 
$A$ is a solution of \ProblemA. 
\end{proof}

By this proposition, in order to solve \ProblemA,   
it suffices to test whether there exists a pair $(A^{\rm blue}, d)$
such that $2 \rho_{A^{\rm blue}} (u) + d(u) \ge 3$ and $\phi_u(A_{\rm ini}) \underset{\mathcal{F}_u}{\leftrightsquigarrow} (A^{\rm blue}, d)$. 
Since $|\mathcal {F}_u| = 2^{O(k)}$, it can be checked in $2^{O(k)} \cdot {\rm poly} (|V|)$ time. 
Note that the elements of $\mathcal {F}_u$ can be computed in $2^{O(k)} \cdot {\rm poly} (|V|)$ time by Lemma~\ref{lem:blue14}. 
Thus, Algorithm~\ref{alg1} solves the problem \ctoc in $2^{O(k)} \cdot {\rm poly} (|V|)$ time.

Using similar arguments as in Theorem~\ref{thm:blue01} we can also solve the \ctoe
version.

\begin{corollary}
\label{cor:blue01}
	\ctoe parameterized by the number $k$ of blue edges can be solved in  time
	$2^{O(k)} \cdot {\rm poly}(|V|)$.
\end{corollary}

\section{Conclusion}
\label{sec:conclusion}

We investigated the parameterized complexity of \ncl for four natural
parameters related to the constraint graph: The number of \andsc/\orsc vertices
of an \andorsc graph and the number of red/blue edges of a general constraint
graph. We give FPT algorithms for the \ctoc and \ctoe version of \ncl for each parameter
and in particular a linear kernel for \ncl parameterized by he number of red
edges. 
An interesting question for future work is whether there is a
polynomial kernel for \ncl parameterized by the number of \orsc vertices or the
number of blue edges.

\bibliography{references}

\begin{thebibliography}{10}

\bibitem{Kempe:19}
Marthe Bonamy, Marc Heinrich, Takehiro Ito, Yusuke Kobayashi, Haruka Mizuta,
  Moritz M{\"{u}}hlenthaler, Akira Suzuki, and Kunihiro Wasa.
\newblock Diameter of colorings under kempe changes.
\newblock In {\em Computing and Combinatorics - 25th International Conference,
  {COCOON} 2019, Xi'an, China, July 29-31, 2019, Proceedings}, pages 52--64,
  2019.
\newblock \href {https://doi.org/10.1007/978-3-030-26176-4\_5}
  {\path{doi:10.1007/978-3-030-26176-4\_5}}.

\bibitem{Paths:19}
Erik~D Demaine, David Eppstein, Adam Hesterberg, Kshitij Jain, Anna Lubiw,
  Ryuhei Uehara, and Yushi Uno.
\newblock Reconfiguring undirected paths.
\newblock In {\em Workshop on Algorithms and Data Structures}, pages 353--365.
  Springer, 2019.

\bibitem{FlakeBlaum}
Gary~William Flake and Eric~B Baum.
\newblock Rush hour is {PSPACE}-complete, or ``why you should generously tip
  parking lot attendants''.
\newblock {\em Theoretical Computer Science}, 270(1-2):895--911, 2002.

\bibitem{Haddadan:16}
Arash Haddadan, Takehiro Ito, Amer~E Mouawad, Naomi Nishimura, Hirotaka Ono,
  Akira Suzuki, and Youcef Tebbal.
\newblock The complexity of dominating set reconfiguration.
\newblock {\em Theoretical Computer Science}, 651:37--49, 2016.

\bibitem{HIZ18}
Tatsuhiko Hatanaka, Takehiro Ito, and Xiao Zhou.
\newblock Complexity of reconfiguration problems for constraint satisfaction.
\newblock {\em CoRR}, abs/1812.10629, 2018.
\newblock URL: \url{http://arxiv.org/abs/1812.10629}, \href
  {http://arxiv.org/abs/1812.10629} {\path{arXiv:1812.10629}}.

\bibitem{HZY:17}
Weihua He, Ziwen Liu, and Chao Yang.
\newblock Snowman is pspace-complete.
\newblock {\em Theoretical Computer Science}, 677:31--40, 2017.

\bibitem{HD:05}
Robert~A. Hearn and Erik~D. Demaine.
\newblock {PSPACE}-completeness of sliding-block puzzles and other problems
  through the nondeterministic constraint logic model of computation.
\newblock {\em Theoretical Computer Science}, 343(1-2):72--96, 2005.
\newblock \doi {10.1016/j.tcs.2005.05.008}.

\bibitem{Acyclic:14}
Hendrik~Jan Hoogeboom, Walter~A Kosters, Jan~N van Rijn, and Jonathan~K Vis.
\newblock Acyclic constraint logic and games.
\newblock {\em ICGA Journal}, 37(1):3--16, 2014.

\bibitem{HH:16}
Robin Houston and Willem Heijltjes.
\newblock Proof equivalence in {MLL} is {PSPACE}-complete.
\newblock {\em Logical Methods in Computer Science}, 12, 2016.

\bibitem{Ito:11}
Takehiro Ito, Erik~D. Demaine, Nicholas~J.A. Harvey, Christos~H. Papadimitriou,
  Martha Sideri, Ryuhei Uehara, and Yushi Uno.
\newblock On the complexity of reconfiguration problems.
\newblock {\em Theoretical Computer Science}, 412(12–14):1054--1065, 2011.
\newblock \href {https://doi.org/10.1016/j.tcs.2010.12.005}
  {\path{doi:10.1016/j.tcs.2010.12.005}}.

\bibitem{Kaminski:12}
Marcin Kami{\'n}ski, Paul Medvedev, and Martin Milani{\v{c}}.
\newblock Complexity of independent set reconfigurability problems.
\newblock {\em Theoretical computer science}, 439:9--15, 2012.

\bibitem{Nishimura:18}
Naomi Nishimura.
\newblock Introduction to reconfiguration.
\newblock {\em Algorithms}, 11(4):52, 2018.

\bibitem{Schrijver:03}
Alexander Schrijver.
\newblock {\em Combinatorial Optimization - Polyhedra and Efficiency},
  volume~24 of {\em Algorithms and Combinatorics}.
\newblock Springer Berlin Heidelberg, 2003.

\bibitem{Heuvel:13}
Jan van~den Heuvel.
\newblock The complexity of change.
\newblock In Simon~R Blackburn, Stefanie Gerke, and Mark Wildon, editors, {\em
  Surveys in Combinatorics 2013}, volume 409. London Mathematical Society
  Lectures Note Series, 2013.

\bibitem{Zanden}
Tom~C van~der Zanden.
\newblock Parameterized complexity of graph constraint logic.
\newblock In {\em 10th International Symposium on Parameterized and Exact
  Computation (IPEC 2015)}. Schloss Dagstuhl-Leibniz-Zentrum fuer Informatik,
  2015.

\bibitem{vdZB:15}
Tom~C van~der Zanden and Hans~L Bodlaender.
\newblock {PSPACE}-completeness of {B}loxorz and of games with 2-buttons.
\newblock In {\em International Conference on Algorithms and Complexity}, pages
  403--415. Springer, 2015.

\end{thebibliography}

\clearpage
\appendix

\section{Proofs Omitted from Section~\ref{sec:OR}}

In this section we prove Theorem~\ref{the:FPT_OR}. To do so, we first formally
define and prove the correctness of the preprocessing step. We then proceed to
give an FPT-reduction from preprocessed instances of \ctoc parameterized by the
number of \orsc vertices to \BCSR{}.

\subsection{Correctness of the Preprocessing}
Let $(\oriG, \oriA_{\ini}, \oriA_{\tar})$ be a given instance of \ncl
such that $\oriG = (\oriV, \oriE)$ is an \andorsc graph.  For each $xy
\in \Blueset{\oriE}$ which is not a loop, we subdivide it by adding a
new vertex $z$; and let $\weig(xz) = \weig(zy) = 2$.  We call the newly
inserted vertex $z$ a \emph{middle vertex} between $x$ and $y$.  Let
$G=(V,E)$ be the resulting graph.  Note that $G$ is not an \andorsc
graph, and $V$ can be partitioned into $\Vand(G)$, $\Vor(G)$ and
$\Vmid(G)$, where $\Vmid(G)$ (or simply $\Vmid$) is the set of middle
vertices in $G$.  Since $\Vor(G) = \Vor(\oriG)$ holds, we will show in
this subsection that solving $(\oriG, \oriA_{\ini}, \oriA_{\tar})$ with
the parameter $|\Vor(\oriG)| \le k$ is equivalent to solving $(G,
A_{\ini}, A_{\tar})$ with the parameter $|\Vor(G)| = |\Vor(\oriG)| \le
k$. 

Consider any orientation $\oriA$ of the original edge set $\oriE$.
Then, we define an orientation $A$ of $E$, as follows: for each $(x,y)
\in \Blueset{\oriA}$ such that $x \neq y$, we delete $(x,y)$ from
$\oriA$, and add two arcs $(x,z)$ and $(z, y)$; let $A$ be the
resulting orientation of $E$.  We observe that the following lemma
holds. 

In particular, we construct two (feasible) orientations $A_{\ini}$ and
$A_{\tar}$ of $E$ which correspond to the original (feasible)
orientations $\oriA_{\ini}$ and $\oriA_{\tar}$ of $\oriE$,
respectively.  In this way, we obtain the instance $(G, A_{\ini},
A_{\tar})$ of \ncl as the result of the preprocessing to $(\oriG,
\oriA_{\ini}, \oriA_{\tar})$. 
\begin{lemma} \label{clm:cons1}
	$A$ is a feasible orientation of $E$ if and only if $\oriA$ is a feasible orientation of $\oriE$.
\end{lemma}
\begin{proof}
	For each middle vertex $z \in \Vmid$ between $x$ and $y$, it holds that $\indegree{\Ared}(z) + 2 \cdot \indegree{\Ablue}(z) = 2$ because $(x,z) \in A$ and $(y,z) \not\in A$. 
	Thus, $A$ is feasible if and only if $\indegree{\Ared}(v) + 2 \cdot \indegree{\Ablue}(v) \ge 2$ for every $v \in V(G) \setminus \Vmid = V(\oriG)$. 
	Therefore, the lemma holds. 
\end{proof}

Let $(\oriG, \oriA_{\ini}, \oriA_{\tar})$ be a given instance of \ncl such that
$\oriG = (\oriV, \oriE)$ is an \andorsc graph.  For each $xy \in
\Blueset{\oriE}$ which is not a loop, we subdivide it by adding a new vertex
$z$; and let $\weig(xz) = \weig(zy) = 2$.  We call the newly inserted vertex
$z$ a \emph{middle vertex} between $x$ and $y$.  Let $G=(V,E)$ be the resulting
graph.  Note that $G$ is not an \andorsc graph, and $V$ can be partitioned into
$\Vand(G)$, $\Vor(G)$ and $\Vmid(G)$, where $\Vmid(G)$ (or simply $\Vmid$) is
the set of middle vertices in $G$.  Since $\Vor(G) = \Vor(\oriG)$ holds, we
will show in this subsection that solving $(\oriG, \oriA_{\ini}, \oriA_{\tar})$
with the parameter $|\Vor(\oriG)| \le k$ is equivalent to solving $(G,
A_{\ini}, A_{\tar})$ with the parameter $|\Vor(G)| = |\Vor(\oriG)| \le k$. 

Consider any orientation $\oriA$ of the original edge set $\oriE$.  Then, we
define an orientation $A$ of $E$, as follows: for each $(x,y) \in
\Blueset{\oriA}$ such that $x \neq y$, we delete $(x,y)$ from $\oriA$, and add
two arcs $(x,z)$ and $(z, y)$; let $A$ be the resulting orientation of $E$.  We
observe that the following lemma holds. 

In particular, we construct two (feasible) orientations $A_{\ini}$ and
$A_{\tar}$ of $E$ which correspond to the original (feasible) orientations
$\oriA_{\ini}$ and $\oriA_{\tar}$ of $\oriE$, respectively.  In this way, we
obtain the instance $(G, A_{\ini}, A_{\tar})$ of \ncl as the result of the
preprocessing to $(\oriG, \oriA_{\ini}, \oriA_{\tar})$.  Lemma~\ref{clm:cons2}
ensures that the preprocessing preserves the reconfigurability. 

\begin{proof}[Proof of Lemma~\ref{clm:cons2}]
	We first prove the if direction.  Assume that $(\oriG, \oriA_{\ini},
	\oriA_{\tar})$ is a $\yes$-instance, and hence there exists a
	reconfiguration sequence $\reconfseq{\oriA_0, \oriA_1, \dots ,
	\oriA_\ell}$ between $\oriA_0 = \oriA_{\ini}$ and $\oriA_{\ell} =
	\oriA_{\tar}$.  For each  $i \in \{0,1,\ldots, \ell\}$, we define the
	arc set $A_i$ by replacing $(x,y) \in \oriA_i$ with two arcs $(x,z),
	(z, y)$ for all middle vertices $z \in \Vmid$ that subdivide $xy \in
	\Blueset{\oriE}$.  Since each $\oriA_i$ is a feasible orientation of
	$\oriE$, Lemma~\ref{clm:cons1} says that $A_i$ is a feasible
	orientation of $E$.  In addition, by the definition, we know that $A_0
	= A_{\ini}$ and $A_\ell = A_{\tar}$.  Then, the following claim proves
	that $(G, A_{\ini}, A_{\tar})$ is a $\yes$-instance. 
	\begin{myclaim} \label{myclaim:preprocess1}
		$A_{i-1} \reachable A_{i}$ holds for each $i \in \{1,2,\ldots, \ell\}$.
	\end{myclaim}
	\begin{proof}[Proof of Claim~\ref{myclaim:preprocess1}]
		If $\oriA_{i-1} = \oriA_i$ holds, then we have $A_{i-1} = A_i$ and hence the claim trivially holds.
		We thus assume that $\oriA_{i-1} \neq \oriA_i$. 
		Since $\oriA_{i-1} \onestep \oriA_i$, there exists an arc $(x,y) \in \oriA_{i-1}$ such that $\oriA_i = \oriA_{i-1} - (x,y) +(y,x)$ where $x \neq y$. 
		If the arc $(x,y) \in \oriA_{i-1}$ is red, then $(x,y) \in A_{i-1}$ and we have $A_{i-1} - (x,y) + (y,x) = A_i$. 
		Therefore, $A_{i-1} \onestep A_i$ holds, and hence we have $A_{i-1} \reachable A_i$.  

		We thus consider the remaining case, that is, the arc $(x,y) \in \oriA_{i-1}$ is blue and $x \neq y$.
		Let $z$ be the middle vertex between $x$ and $y$. 
		Then, we know that $A_{i-1} \setminus A_i = \{(x,z), (z,y)\}$ and $A_i \setminus A_{i-1} = \{(z,x), (y,z)\}$.
		Let $A := A_{i-1} - (z,y) + (y,z)$, then $A_{i-1} \onestep A$ holds.
		We now prove that $\reconfseq{A_{i-1}, A, A_i}$ is a reconfiguration sequence between $A_{i-1}$ and $A_i$. 
		Note that $A \onestep A_i$ holds since $A_i = A - (x,z) + (z,x)$. 
		Therefore, it suffices to show that $A$ is a feasible orientation of $E$. 
		To see this, we observe that for each $v \in V(G)$, 
		\[
			\indegree{\Ared}(v) + 2 \cdot \indegree{\Ablue}(v) = \begin{cases}
				4 & \mbox{if $v=z$};\\
				\indegree{\Ared_i}(v) + 2 \cdot \indegree{\Ablue_i}(v) & \mbox{if $v=y$}; \\
				\indegree{\Ared_{i-1}}(v) + 2 \cdot \indegree{\Ablue_{i-1}}(v) & \mbox{otherwise}.
			\end{cases}
		\]
		Since both $A_{i-1}$ and $A_i$ are feasible orientations of $E$, it thus holds that $\indegree{\Ared}(v) + 2 \cdot \indegree{\Ablue}(v) \ge 2$ for all $v \in V(G)$.
		Therefore, $A$ is also feasible, and hence we have $A_{i-1} \reachable A_i$. 
	\end{proof}
	This completes the proof of the if direction. 
	\medskip

	We then prove the only-if direction.  Assume that $(G, A_{\ini},
	A_{\tar})$ is a $\yes$-instance, and hence there exists a
	reconfiguration sequence $\reconfseq{A_0, A_1, \dots , A_\ell}$ between
	$A_0 = A_{\ini}$ and $A_{\ell} = A_{\tar}$.  For each $i \in
	\{0,1,\ldots,\ell\}$, we define an orientation $\oriA_i$ of $\oriE$
	from $A_i$, as follows: for each middle vertex $z$ in $G$ which
	subdivides $xy \in \Blueset{\oriE}$, 
	\begin{itemize}
		\item if $(z,y) \in A_i$, then we replace two arcs $(x,z), (z,y) \in A_i$ with a single arc $(x,y)$; 
		\item otherwise we replace two arcs containing $z$ with a single arc $(y,x)$;
	\end{itemize}
	let $\oriA_i$ be the resulting arc set.  Note that both $\oriA_0 =
	\oriA_{\ini}$ and $\oriA_{\ell} = \oriA_{\tar}$ hold.  We observe that
	each $\oriA_i$ is a feasible orientation of $\oriE$, because $A_i$ is a
	feasible orientation of $E$ and it holds that 
	\[
		\indegree{\Redset{\oriA}_i}(v) + 2 \cdot \indegree{\Blueset{\oriA}_i}(v) \ge \indegree{\Ared_i}(v) + 2 \cdot \indegree{\Ablue_i}(v)
	\]
	for each $v \in V(\oriG) = V(G)\setminus \Vmid$.  Then, we have the
	following claim. 
	\begin{myclaim} \label{myclaim:preprocess2}
		For each $i \in \{1,2,\ldots, \ell\}$, it holds that $\oriA_{i-1} \onestep \oriA_{i}$.
	\end{myclaim}
	\begin{proof}[Proof of Claim~\ref{myclaim:preprocess2}]
		If $A_{i-1} = A_i$ holds, then we have $\oriA_{i-1} = \oriA_i$ and hence the claim holds.
		We thus assume that $A_{i-1} \neq A_i$.
		Since $A_{i-1} \onestep A_i$, there exists an arc $(u,v) \in A_{i-1}$ such that $A_i = A_{i-1} - (u,v) +(u,v)$ where $u \neq v$.  
		If the arc $(u,v) \in A_{i-1}$ is red, then $(u,v) \in \oriA_{i-1}$ and we have $\oriA_{i-1} - (u,v) + (v,u) = \oriA_i$. 
		Therefore, $\oriA_{i-1} \onestep \oriA_i$ holds.  

		We thus consider the remaining case, that is, the arc $(u,v) \in A_{i-1}$ is blue and $u \neq v$.
		Then, we know that either $u$ or $v$ is a middle vertex $z$ which was inserted to a blue edge $xy \in \Blueset{\oriE}$. 
		If $(u,v) = (z,y) \in A_{i-1}$, then $(y,z) \in A_i$ and hence we have $(x,y) \in \oriA_{i-1}$ and $\oriA_{i-1} - (x,y) + (y,x) = \oriA_i$; it thus holds that $\oriA_{i-1} \onestep \oriA_i$.  
		If $(u,v) \neq (z,y)$, then we have $\oriA_{i-1} = \oriA_i$ and hence  $\oriA_{i-1} \onestep \oriA_i$.
	\end{proof}

	By Claim~\ref{myclaim:preprocess2} we can obtain a reconfiguration
	sequence between $\oriA_{\ini} = \oriA_0$ and $\oriA_{\tar} =
	\oriA_\ell$ as a sub-sequence of $\reconfseq{\oriA_0, \oriA_1, \ldots,
	\oriA_\ell}$ by ignoring repetitions of the same orientations.  Thus,
	$(\oriG, \oriA_{\ini}, \oriA_{\tar})$ is a $\yes$-instance.  This
	completes the proof of the only-if direction. 
\end{proof}

In this way, to prove Theorem~\ref{the:FPT_OR}, it suffices to solve the
instance $(G, A_{\ini}, A_{\tar})$ with the parameter $|\Vor(G)| =
|\Vor(\oriG)| \le k$.  Recall that $V(G)$ can be partitioned into three subsets
$\Vand(G)$, $\Vor(G)$ and $\Vmid(G)$.  By the construction of $G$, observe that
$G$ has no multiple blue edges.  In addition, no edge in $G$ joins an \andsc
vertex and an \orsc vertex, and hence we can partition $E$ into two subsets
$\Eand$ and $\Eor$, defined as follows: $\Eand$ is the set of edges in $G$ that
are incident to \andsc vertices in $V(G)$; and $\Eor$ is the set of edges in
$G$ that are incident to \orsc vertices in $V(G)$. 

\subsection{FPT-reduction to \BCSR}

In this subsection, we construct an FPT-reduction to \BCSR.
Recall that \BCSR can be solved in time $O^*(\doms^{O(\NB)})$, where $\doms :=
\max_{x \in X} |\Dom{x}|$ and $\NB$ is the number of non-Boolean variables in
$X$~\cite[Theorem~18]{HIZ18}.  Therefore, as a proof of
Theorem~\ref{the:FPT_OR}, we construct an FPT-reduction from a preprocessed
instance $(G, A_{\ini}, A_{\tar})$ of \ncl with the parameter $|\Vor(G)| \le k$
to an instance $(H,\Dommap,\Cnsmap,\mapf_\ini,\mapf_\tar)$ of \BCSR such that
both $\doms$ and $\NB$ are bounded by some computable functions depending only
on $k$.

We first construct the set $\Var$ of variables, as follows:
\begin{itemize}
	\item for each $v \in \Vor(G)$, we introduce a variable $x_v$, called an \emph{\orsc variable}; and 
	\item for each $e\in \Eand$, we introduce a variable $x_e$, called an \emph{edge variable}.
\end{itemize}

We then construct the domain $\Dom{x}$ for each $x \in \Var$, as follows:
\begin{itemize}
	\item If $x$ is an \orsc variable $x_v$ for $v \in \Vor(G)$, then we consider the following two cases: 
		\begin{itemize}
			\item Consider the case where there is a loop $vv \in \Eor$.
				Since $v$ is an \orsc vertex, it has exactly one blue edge $vv' \in \Eor$ such that $v'$ is a middle vertex.
				In this case, let $\Dom{x_v} := \{\emptyset, \{v'\}\}$. 
				We regard that assigning $\{v'\} \in \Dom{x_v}$ to $x_v$ corresponds to directing the edge $vv'$ as $(v,v')$, while $\emptyset \in \Dom{x_v}$ to $(v',v)$. 
				Note that the loop $vv$ has only one possible direction, and any orientation of $E$ contains the arc $(v,v)$.

			\item Consider the other case, that is, $vv \not\in \Eor$.
				Since $G$ has no multiple blue edges, $v$ has three distinct neighbors (middle vertices), say $v_1$, $v_2$ and $v_3$, in $G$.
				In this case, let $\Dom{x_v} := \{ \emptyset, \{v_1 \}, \{v_2 \}, \{v_3 \}, \{v_1, v_2 \}, \{v_2, v_3 \}, \{v_3 , v_1 \} \}$.
				We regard that assigning a set $S \in \Dom{x_v}$ to $x_v$ corresponds to directing the three edges $vv_1$, $vv_2$ and $vv_3$ as follows: $(v,v')$ if $v' \in S$, and $(v',v)$ if $v' \in \{v_1, v_2, v_3\} \setminus S$. 
		\end{itemize} 
	\item If $x$ is an edge variable $x_e$ for $e = uv \in \Eand$, then let $\Dom{x_e} := \{ \{u\}, \{v\} \}$. 
		We regard that assigning $\{u\} \in \Dom{x_e}$ to $x_e$ corresponds to directing $uv$ as $(v,u)$, while $\{v\} \in \Dom{x_e}$ to $(u,v)$.
\end{itemize}
Since $E$ can be partitioned into $\Eor$ and $\Eand$, any solution $\mapf$ of $H$ defines an orientation of $E$.
Conversely, any orientation of $E$ defines a solution $\mapf$ of $H$. 
We note that $\doms = \max_{x \in X} |\Dom{x}| = 7$. 
Furthermore, notice that only \orsc variables $x_v$ without loops are non-Boolean variables, and the other variables are Boolean variables.
Therefore, $\NB \le |\Vor(G)| \le k$, where $\NB$ is the number of non-Boolean variables in $X$. 

We finally construct the set of constraints, which also defines the set of edges in $H$.
Our aim here is to ensure that a solution $\mapf$ of $G$ is proper if and only if the corresponding orientation $A_{\mapf}$ of $E$ is feasible. 
By the construction of domains above, we know that $\indegree{\Ared_{\mapf}}(v) + 2 \cdot \indegree{\Ablue_{\mapf}}(v) \ge 2$ holds for each $v \in \Vor(G)$.
Therefore, we construct three types of constraints for middle vertices and \andsc vertices, as follows:
\begin{description}
	\item[Type 1:] Constraints for middle vertices.

		Let $v$ be a middle vertex between two vertices $v_1$ and $v_2$.
		Since both $v_1 v$ and $v v_2$ are blue edges, $\indegree{\Ared_{\mapf}}(v) + 2 \cdot \indegree{\Ablue_{\mapf}}(v) \ge 2$ holds if and only if $(v_1,v) \in A_{\mapf}$ or $(v_2,v) \in A_{\mapf}$ hold. 
		For each $i \in \{1,2\}$, let
		\[
			x_i = \begin{cases}
				x_{v_i} & \mbox{if $v_i$ is an \orsc vertex};\\
				x_{vv_i} & \mbox{otherwise}.
			\end{cases}
		\]
		Then, we let $\Cns{x_1 x_2} := \{ S_1 S_2 \in \Dom{x_1} \times \Dom{x_2} \midcolon v \in S_1 \mbox{ or } v \in S_2 \}$.

	\item[Type 2-1:] Constraints for \andsc vertices having loops.

		Let $v$ be an \andsc vertex having a loop $vv$. 
		Since $v$ is an \andsc vertex, we know that $vv$ must be red, and the remaining edge $vv_3 \in \Eand$ is blue where $v_3$ is a middle vertex. 
		Then, $\indegree{\Ared_{\mapf}}(v) + 2 \cdot \indegree{\Ablue_{\mapf}}(v) \ge 2$ holds if and only if $(v_3,v) \in A_{\mapf}$.
		Since $vv$ and $vv_3$ are in $\Eand$, there are corresponding edge variables $x_{vv}$ and $x_{vv_3}$.
		Then, we let $\Cns{x_{vv} x_{vv_3}}:=\{ S S_3 \in \Dom{x_{vv}} \times \Dom{x_{vv_3}} \midcolon v \in S_3 \}$.

	\item[Type 2-2:] Constraints for \andsc vertices having no loop.

		Let $v$ be an \andsc vertex, and let $vv_1$, $vv_2$, $vv_3$ be three (distinct) edges incident to $v$ such that $vv_1$ and $vv_2$ are red, and $vv_3$ is blue; it may hold that $v_1 =v_2$.
		Then, $\indegree{\Ared_{\mapf}}(v) + 2 \cdot \indegree{\Ablue_{\mapf}}(v) \ge 2$ holds if and only if $A_{\mapf}$ satisfies both the following two conditions:
		\begin{enumerate}
			\item it holds that $(v_1,v) \in A_{\mapf}$ or $(v_3,v) \in A_{\mapf}$; and 
			\item it holds that $(v_2,v) \in A_{\mapf}$ or $(v_3,v) \in A_{\mapf}$.
		\end{enumerate}
		Since $vv_1$, $vv_2$ and $vv_3$ are in $\Eand$, there are corresponding edge variables $x_{vv_1}$, $x_{vv_2}$ and $x_{vv_3}$.
		Then, we let $\Cns{x_{vv_1} x_{vv_3}}:=\{ S_1 S_3 \in \Dom{x_{vv_1}} \times \Dom{x_{vv_3}} \midcolon v \in S_1 \mbox{ or } v \in S_3 \}$, and $\Cns{x_{vv_2} x_{vv_3}}:=\{ S_2 S_3 \in \Dom{x_{vv_2}} \times \Dom{x_{vv_3}} \midcolon v \in S_2 \mbox{ or } v \in S_3 \}$.
\end{description}
By the construction of constraints above, we know that a solution $\mapf$ of
$G$ is proper if and only if the corresponding orientation $A_{\mapf}$ of $E$
is feasible.  Therefore, we can define proper solutions $\mapf_{\ini}$ and
$\mapf_{\tar}$ of $H$ which correspond to feasible orientations $A_{\ini}$ and
$A_{\tar}$ of $E$, respectively. 

In this way, from a preprocessed instance $(G, A_{\ini}, A_{\tar})$ of \ncl
with the parameter $|\Vor(G)| \le k$, we have constructed the corresponding
instance $(H,\Dommap,\Cnsmap,\mapf_\ini,\mapf_\tar)$ of \BCSR such that $\doms
= \max_{x \in X} |\Dom{x}| = 7$ and $\NB \le |\Vor(G)| \le k$.  In addition, we
have shown that there is a one-to-one correspondence between proper solutions
of $H$ and feasible orientations of $E$.  Since \BCSR can be solved in time
$O^*(\doms^{O(\NB)})$~\cite{HIZ18}, the following lemma completes the proof of
Theorem~\ref{the:FPT_OR} for \ctoc.
\begin{lemma} \label{lem:equivalence}
	$(G, A_{\ini}, A_{\tar})$ is a $\yes$-instance of \ncl if and only if $(H,\Dommap,\Cnsmap,\mapf_{\ini},\mapf_{\tar})$ is a $\yes$-instance of \BCSR.
\end{lemma}
\begin{proof}
	We first prove the only-if direction.
	Assume that $(G, A_{\ini}, A_{\tar})$ is a $\yes$-instance, and hence there exists a reconfiguration sequence $\reconfseq{A_0, A_1, \dots , A_\ell}$ of feasible orientations of $E$ between $A_0 = A_{\ini}$ and $A_\ell = A_{\tar}$.
	For each $i \in \{0,1,\ldots, \ell \}$, let $\mapf_i$ be the proper solution of $H$ defined by $A_i$.
	Then, we know that $\mapf_0 = \mapf_\ini$ and $\mapf_\ell = \mapf_\tar$ hold. 
	To show that $(H,\Dommap,\Cnsmap,\mapf_{\ini},\mapf_{\tar})$ is a $\yes$-instance, it thus suffices to prove that $\mapf_{i-1} \onestep \mapf_i$ holds for each $i \in \{1,2,\ldots, \ell \}$.
	Since $A_{i-1} \onestep A_i$, there exists an arc $(u,v) \in A_{i-1}$ such that $A_i = A_{i-1} - (u,v) +(v,u)$.  
	We know that the edge $uv$ in $G$ is either in $\Eand$ or in $\Eor$.  
	If $uv \in \Eand$, then we have $|\{x \in \Var \midcolon \mapf_{i-1}(x) \neq \mapf_i(x) \}| = |\{x_{uv}\}| = 1$.
	Otherwise (i.e., if $uv \in \Eor$), then we have $|\{x \in \Var \midcolon \mapf_{i-1}(x) \neq \mapf_i(x)\}| = |\{x_v\}| = 1$ where we assume without loss of generality that $v$ is an \orsc vertex and $u$ is a middle vertex.
	Therefore, $\mapf_{i-1} \onestep \mapf_i$ holds for both cases, as claimed. 
	This completes the proof of the only-if direction. 

	We then prove the if direction.
	Assume that $(H,\Dommap,\Cnsmap,\mapf_{\ini},\mapf_{\tar})$ is a $\yes$-instance, and hence there exists a sequence $\reconfseq{\mapf_0, \mapf_1, \ldots , \mapf_\ell}$ of proper solutions of $H$ such that $\mapf_0 = \mapf_\ini$, $\mapf_\ell = \mapf_\tar$, and $\mapf_{i-1} \onestep \mapf_i$ holds for each $i \in \{1,2,\ldots, \ell \}$.
	For each $i \in \{0,1,\ldots, \ell \}$, let $A_i$ be the feasible orientation of $E$ defined by $\mapf_i$.
	Then, we know that $A_0 = A_\ini$ and $A_\ell = A_\tar$ hold. 
	To show that $(G, A_{\ini}, A_{\tar})$ is a $\yes$-instance, we thus prove that $A_{i-1} \reachable A_i$ holds for each $i \in \{1,2,\ldots, \ell \}$. 
	Since $\mapf_{i-1} \onestep \mapf_i$, there exists exactly one variable $x \in \Var$ such that $\mapf_{i-1}(x) \neq \mapf_i(x)$. 
	Then, we consider the following three cases.
	\begin{description}
		\item[Case 1.] $x$ is an edge variable $x_{e}$ for $e \in \Eand$.

			In this case, we know that the difference between $A_{i-1}$ and $A_i$ is only the direction of $e$. 
			Therefore, $A_i$ can be obtained from $A_{i-1}$ by reversing the direction of $e$, and hence we have $A_{i-1} \onestep A_i$.
			Thus, $A_{i-1} \reachable A_i$ holds.

		\item[Case 2-1.] $x$ is an \orsc variable $x_v$ for $v \in \Vor(G)$ having a loop $vv \in \Eor$. 

			In this case, we know that the difference between $A_{i-1}$ and $A_i$ is only the direction of $vv' \in \Eor$, where $v'$ is the (unique) middle vertex adjacent to $v$. 
			Therefore, we have $A_{i-1} \onestep A_i$, and hence $A_{i-1} \reachable A_i$ holds.

		\item[Case 2-2.]	$x$ is an \orsc variable $x_v$ for $v \in \Vor(G)$ having no loop. 

			Let $v_1$, $v_2$, $v_3$ be three (distinct) middle vertices adjacent to $v$.
			Recall that a solution $\mapf$ of $H$ defines the directions of three edges $vv_1$, $vv_2$ and $vv_3$ in the corresponding orientation of $E$, as follows:
			$(v,v')$ if $v' \in \mapf(v)$, and $(v',v)$ if $v' \in \{v_1, v_2, v_3\} \setminus \mapf(v)$. 
			Then, we construct a sequence of orientations of $E$ between $A_{i-1}$ and $A_i$, as follows: 
			\begin{enumerate}
				\item for each $v' \in \mapf_{i-1}(v) \setminus \mapf_i(v)$, reverse the direction of $vv'$ from $(v, v')$ to $(v', v)$ one by one; and 
				\item for each $v' \in \mapf_i(v) \setminus \mapf_{i-1}(v)$, reverse the direction of $vv'$ from $(v', v)$ to $(v, v')$ one by one. 
			\end{enumerate}
			Let $\reconfseq{A^0, A^1, \ldots, A^q}$ be the sequence of orientations of $G$ defined as above, where $A^0 := A_{i-1}$ and $q := |\mapf_{i-1}(v) \symdiff \mapf_i(v)|$.
			By the construction of the sequence, we know that $A^q = A_i$ and $A^{j-1} \onestep A^j$ for each $j \in \{1,2,\ldots, q\}$. 
			Furthermore, 
			all orientations $A^j$ are feasible.
			Thus, $\reconfseq{A^0, A^1, \ldots, A^q}$ is a reconfiguration sequence between $A^0 = A_{i-1}$ and $A^q = A_i$, and hence we have $A_{i-1} \reachable A_i$. 
	\end{description}
	This completes the proof of the if direction.
\end{proof}

It remains to prove the statement for \ctoe.
For the \ctoc case, we give a reduction from \ncl to \BCSR.
We use the same reduction for \ctoe case.
While an FPT-algorithm for \ctoc of \BCSR is given in~\cite{HIZ18}, that for \ctoe is not.
However, we can simply improve the FPT-algorithm for \ctoc to \ctoe as follows:
In the FPT-algorithm of~\cite{HIZ18}, the authors first construct a contracted solution graph (CSG).
Then they determine whether there is a path on CSG between two nodes corresponding to the initial and the target solutions.
Since the size of CSG is FPT-size, the algorithm takes only FPT-time.
Then for the \ctoe case, it is enough to determine whether there is a path on CSG between two nodes corresponding to the initial solution and any solution we wish.

In the remaining part of this proof, we show that the edge we wish to reverse in \ncl is which variable in \BCSR.
In the preprocessing of our reduction, we subdivide each blue edge that is not a loop into two blue edges.
Let $G$ be the graph before the preprocessing and $\oriG$ be a graph after preprocessing.
Let $(u^*, v^*)$ be the orientation in $A_{\ini}$ of $G$ that we wish to reverse.
Note that the edge $u^* v^*$ may be subdivided into $u^* z^*$ and $z^* v^*$ in the preprocessing.
In other words, reversing $(u^*, v^*)$ in $G$ corresponds to reversing both of $(u^*, z^*)$ and $(z^*, v^*)$ in $\oriG$.
However, since $z^* v^*$ is a blue edge, after $(z^*, v^*)$ is reversed, $z^*$ has enough in-weight and we can always reverse $(u^*, z^*)$.
Therefore, reversing $(u^*, v^*)$ in $G$ corresponds to reversing only $(z^*, v^*)$ in $\oriG$.

Let $e^*$ be the edge of $\oriG$ that we wish to reverse.
If $e^*$ is incident to an \orsc vertex, the three edges including $e^*$ in \ncl correspond to one non-Boolean variable in \BCSR, say $v$.
While $v$ can be assigned seven values, four of them correspond to incoming direction of $e^*$, and three of them correspond to outgoing direction of $e^*$.
Therefore, we need to search for a path between a node corresponding to the initial orientation
and any node such that the value of $v$ corresponds to a different direction of $e^*$ from the initial one.

If $e^*$ is incident to an \andsc vertex, $e^*$ in \ncl corresponds to a Boolean variable in \BCSR.
Therefore, reversing the direction of $e^*$ corresponds to changing a value of this variable.
In this case, a node corresponds to the initial orientation and a node corresponds to an orientation we wish to obtain. However, some nodes might be contracted in CSG.
Before some nodes are contracted, some variables of which we cannot change the value at all
are deleted.
Therefore, if $e^*$ corresponds to such a variable, we answer NO.
Otherwise, we only check whether there exists at least one feasible solution in \BCSR such that the variable has a different value from the initial one.
Since contracted nodes in CSG are always connected, if there exists such a feasible solution, we answer YES, and otherwise NO.

By above discussion, our reduction for \ctoc also works for \ctoe.

\section{Proofs Omitted from Section~\ref{sec:red}}

\begin{proof}[Proof of Proposition~\ref{prop:reduction}]
  We prove that the four reduction rules are safe one by one. We first show
  that Reduction Rule~\ref{red-rule:onecycle} is safe.

  Let $C$ be a component of $G$ that is a blue chordless cycle.  Observe that
  no arc on $C$ can be reversed. Therefore, if $A_{\rm ini}$ and $A_{\rm tar}$
  disagree on $C$, then we have a \NO instance.  On the other hand, if $A_{\rm
  ini}$ and $A_{\rm tar}$ agree on $C$ then we may remove the component $C$
  from $G$ and continue.

  We now show that Reduction Rule~\ref{red-rule:twocycles} is safe.
  Let $C$ be a component of $(V, E^{\rm blue})$ containing at least two cycles $K_1$ and $K_2$ and let
  $k_C$ be the number of red vertices of $C$. Without loss of generality we
  assume that no proper subset of $E(K_1) \cup E(K_2)$ contains two distinct
  cycles.
  Let $G_C := (V(C), E^{\rm blue} \cap E(C))$ be the graph induced by the vertices of the 
  component $C$ of the blue subgraph of $G$ and let $A_{\rm ini}$ and $A_{\rm tar}$ be the
  start and target orientations of $E$.
  We first prove the following two claims.
  \setcounter{myclaim}{0}
  \begin{myclaim}
    There is a feasible orientation $A^\circ_{\rm ini}$ of the constraint graph $G$ with
    the following properties:
    \begin{enumerate}
      \label{myclaim:alpha1-alpha2}
      \item The orientations $A_{\rm ini}$ and $A^\circ_{\rm ini}$ agree on $E - E(C)$,
      \item $K_1$ and $K_2$ are oriented cycles with respect to $A^\circ_{\rm ini}$, 
      \item each vertex $v$ of $C$ has at least one edge in $E(C)$ oriented towards $v$ by $A^\circ_{\rm ini}$, and
      \item there is a transformation from $A_{\rm ini}$ to $A^\circ_{\rm ini}$. 
    \end{enumerate} 
  \end{myclaim}
  
    \begin{proof}[Proof of Claim \ref{myclaim:alpha1-alpha2}]
    If $K_1$ or $K_2$ is a directed cycle with respect to the orientation
    $A_{\rm ini}$, we leave the orientation of the cycle as it is.
    Otherwise, assume without loss of generality that $K_1$ is not a directed cycle with respect to $A_{\rm ini}$.
    Then 
    there is at least one vertex $v$ of $K_1$ having two blue edges of $K_1$ oriented towards $v$ by $A_{\rm ini}$. 
    Reversing one of the two arcs yields a feasible orientation. 
    After this step one of the neighbors of $v$, say $w$, has one additional incoming edge of $K_1$.
    If the remaining edge of $w$ on $K_1$ is an incoming edge, we reverse this edge.
    Otherwise we leave it as it is.
    By performing these
    steps in a consistent manner we obtain an orientation such that
    $K_1$ is an oriented cycle.   
    If $K_1$ and $K_2$ intersect, 
    due to the minimality of $E(K_1) \cup E(K_2)$, they intersect in a path $P$.
    In a similar way as above, we turn $K_2$ into an
    oriented cycle. In the case that $K_1$ and $K_2$ intersect we orient the edges of $K_2$, 
    such that the orientation is consistent with that of $P$.
    Let $A'_{\rm ini}$ be the resulting feasible orientation.  
    Consider a spanning tree $T$ of $G_C$. Since each vertex of $K_1$ and
    $K_2$ has in-degree at least one with respect to $A'_{\rm ini}$, we may
    (iteratively) direct each edge in $E(T) - (E(K) \cup E(K'))$ away from
    $K_1$.  Let $A^\circ_{\rm ini}$ be the resulting feasible orientation.  Observe that
    each vertex in $G_C$ has in-degree at least one with respect to $A^\circ_{\rm ini}$
    and that only the orientation of edges in $G_C$ were changed. Hence the three
    claimed properties are satisfied by $A^\circ_{\rm ini}$. Furthermore, there is a
    transformation from $A_{\rm ini}$ to $A^\circ_{\rm ini}$.
  \end{proof}

  \begin{myclaim}
    \label{myclaim:beta1-beta2}
    There is a feasible orientation $A^\circ_{\rm tar}$ of the constraint graph $G$ with
    the following properties:
    \begin{enumerate}
      \item The orientations $A_{\rm tar}$ and $A^\circ_{\rm tar}$ agree on $E - E(C)$, 
      \item $A^\circ_{\rm ini}$ and $A^\circ_{\rm tar}$ agree on $E(C)$, and
      \item there is a transformation from $A_{\rm tar}$ to $A^\circ_{\rm tar}$.
    \end{enumerate} 
  \end{myclaim}
  
    \begin{proof}[Proof of Claim \ref{myclaim:beta1-beta2}] 
    We distinguish the two cases that $K_1$ and $K_2$ are disjoint or not.
    Let us first assume that $K_1$ and $K_2$ are disjoint.
    We first show that we can transform $A_{\rm tar}$ into a feasible orientation 
    that agrees with $A^\circ_{\rm ini}$ on $K_1$ and $K_2$. 
    If $K_1$ has a vertex that has at least two incoming edges,
    we apply the same procedure as in the proof of Claim \ref{myclaim:alpha1-alpha2}.
    Else, $K_1$ is a directed cycle (but possibly not directed as in $A^\circ_{\rm ini}$).
    Since $C$ is connected, $K_1$ and $K_2$ are connected by some path $P$.
    Since $K_1$ is a directed cycle, we can direct $P$ away from $K_1$ towards $K_2$.
    Hence, there is at least one vertex in $V(K_2)$ that
    has at least two incoming edges. 
    By the same steps of the proof of Claim \ref{myclaim:alpha1-alpha2} we 
    may obtain a feasible orientation that agrees with
    $A^\circ_{\rm ini}$ on $E(K_2)$. 
    By reversing $P$ and applying the same steps for $K_1$, we obtain an orientation
    $A_{\rm tar}'$ such that $A^\circ_{\rm ini}$ and $A'_{\rm tar}$ agree on $K_1$ and $K_2$.
    We consider the same spanning tree $T$ as in the proof of Claim
    \ref{myclaim:alpha1-alpha2} and (iteratively) direct all edges of $T$ away
    from $K_1$. Let $A''_{\rm tar}$ be the resulting feasible orientation.  Since each
    vertex of $G_C$ has at least one incoming arc with respect to $A''_{\rm tar}$, we
    can direct the remaining edges $E(G_C) - (E(K_1) \cup E(K_2) \cup E(T))$ as
    in $A^\circ_{\rm ini}$. Let the resulting feasible orientation be $A^\circ_{\rm tar}$.
    Observe that in the steps above only the orientation of edges of $G_C$ were
    changed, thus the first property of Claim \ref{myclaim:beta1-beta2} holds.
    Also observe that $A^\circ_{\rm tar}$ and $A^\circ_{\rm ini}$ agree on $E(G_C)$. Thus, Property
    2 also holds.
    
    It remains to consider the case that $K_1$ and $K_2$ are not disjoint.  To obtain a
    feasible orientation $A^\circ_{\rm tar}$ with the desired properties, simply apply
    the steps in the proof of Claim \ref{myclaim:alpha1-alpha2} to $A_{\rm tar}$.  
  \end{proof}
  It follows that there is a transformation
  from $A_{\rm ini}$ to $A_{\rm tar}$ if and only if there is a transformation from
  $A^\circ_{\rm ini}$ to $A^\circ_{\rm tar}$.
  Let $G'$ be the constraint graph obtained from $G$ by deleting the blue
  vertices of $C$ and connecting each red vertex of $C$ with a copy of the
  gadget shown in \figurename~\ref{fig:felix}. 
  We refer to the copies of the gadget as $G_1, G_2, \ldots, G_{k_C}$.
  Let $A^{\rm new}_{\rm ini}$ (resp.  $A^{\rm new}_{\rm tar}$) be an orientation of
  $E(G')$ such that $A^\circ_{\rm ini}$ and $A^{\rm new}_{\rm ini}$ (resp. $A^\circ_{\rm tar}$ and $A^{\rm new}_{\rm tar}$)
  agree on $E(G') - \bigcup_{1 \leq i \leq k_C} E(G_i)$. Furthermore, at each
  gadget $G_i$, $1 \leq i \leq k_C$, the orientations $A^{\rm new}_{\rm ini}$ and
  $A^{\rm new}_{\rm tar}$ are as shown in \figurename~\ref{fig:felix}.
  Observe that by construction each red vertex of $G'$ has at least one
  incoming blue arc from the gadget shown in \figurename~\ref{fig:felix}, so
  the orientations $A^{\rm new}_{\rm ini}$ and $A^{\rm new}_{\rm tar}$ are feasible. 
  
  In order to show that Rule~\ref{red-rule:twocycles} is safe it remains to
  prove that there is a transformation from $A^\circ_{\rm ini}$ to $A^\circ_{\rm tar}$ if and only
  if there is a transformation from $A^{\rm new}_{\rm ini}$ to $A^{\rm new}_{\rm tar}$. So first suppose
  that there is a transformation from $A^\circ_{\rm ini}$ to $A^\circ_{\rm tar}$. Then we obtain a
  transformation from $A^{\rm new}_{\rm ini}$ to $A^{\rm new}_{\rm tar}$ by skipping all the moves that
  change the orientation of an edge in $E(C)$.
  On the other hand, from a transformation from $A^{\rm new}_{\rm ini}$ to $A^{\rm new}_{\rm tar}$ we
  obtain a transformation from $A^\circ_{\rm ini}$ to $A^\circ_{\rm tar}$ by ignoring all the
  moves that change the orientation of an edge of one of the gadgets $G_i$, $1
  \leq i \leq k_C$. Therefore, Rule~\ref{red-rule:twocycles} is safe.

  We now prove that Reduction Rule~\ref{red-rule:degree1} is safe.  Consider a
  blue vertex $v$ of degree one in $G$. Observe that
  Rule~\ref{red-rule:degree1} is safe since in any feasible orientation of
  $E(G)$, the blue edge incident to $v$ is oriented towards $v$.

  It remains to prove the safeness of Reduction Rule~\ref{red-rule:2-path}.
  Let $v$ be a blue vertex of degree 2 in $G$ and let $u$ and $w$ be the
  neighbors of $v$, such that  $uw \notin E(G)$ and let $A$ be a
  feasible orientation of $E$. Since $v$ is a blue vertex of degree 2, there
  are at most three possible orientation of the edges $uv$ and $vw$:  $uv$ and
  $vw$, $wv$ and $vu$, and $uv$ and $wv$.  Suppose we obtain the graph $G'$ by
  replacing the vertex $v$ by a new blue edge $uw$.  We obtain from a feasible
  orientation of $G$ a feasible orientation of $G'$ by orienting $uw$, such
  that the in-weights of $u$ and $w$ in $G'$ are at least the in-weights of $u$
  and $w$ in $G$. Since this is always possible, we obtain from a
  transformation between two feasible orientations of $G$ a transformation
  between two corresponding feasible orientations of $G'$ and vice versa.
  Hence, Rule~\ref{red-rule:2-path} is safe if $u$ and $w$ are non-adjacent in
  $G$.
\end{proof}

\begin{proof}[Proof of Proposition~\ref{prop:running-time:rules}] 
To apply rules~\ref{red-rule:onecycle} and \ref{red-rule:twocycles} we first
run a DFS on $G_B = (V, E^{\rm blue})$. For each component, we check if it is a
cycle, or whether it contains at least two cycles (this can be done by counting).
This takes time $O(|V| + |E|)$.
Additionally it is easy to see that also the rules \ref{red-rule:degree1} and \ref{red-rule:2-path}
can be applied in $O(|V| + |E|)$.
Since in each iteration we delete at least one vertex of the constraint graph, the total running time is
$O(|V| \cdot (|V| + |E|))$.
\end{proof}

\begin{proof}[Proof of Corollary~\ref{cor:red-edges-kernel:c2e}]
We also work with the four rules introduced earlier, but modify Rule~\ref{red-rule:twocycles} slightly.
The key difference is that we only have a starting configuration $A_{\rm ini}$ 
and ask whether a certain edge $e^*$ can
be reversed. Thus, starting from the initial orientation we obtain an instance of size $O(k)$,
but have to keep track of $e^*$.

If $e^*$ is red, then we simply apply Rules \ref{red-rule:twocycles}, \ref{red-rule:degree1} and \ref{red-rule:2-path} as before.
After applying the reduction rules until no longer possible, we observe that the resulting instance
has size $O(k)$ and we can check whether the current orientation is connected
to an orientation in which $e^*$ is reversed.

Thus, we now assume that $e^*$ is blue.
We first consider Rule \ref{red-rule:twocycles} and show how we keep track of $e^*$.
Let $C$ be a connected component of $(V, \Eblue)$ containing at least two cycles.
If $e^* \notin E(C)$, we simply do the same as before.
Observe that this preserves \YES and \NO instances.
We now assume that $e^* \in E(C)$.
In the proof of Claim~\ref{myclaim:alpha1-alpha2} we showed that there is an
orientation $A^\circ_{\rm ini}$ such that
\begin{itemize}
\item[a)] we can choose the orientation of each cycle,
\item[b)] we can choose the orientation of each path connecting two cycles,
\item[c)] all other edges are oriented away from the cycles and
\item[d)] $A_{\rm ini} \leftrightsquigarrow A^\circ_{\rm ini}$.
\end{itemize} 
Thus if $e^*$ is an edge satisfying a) or b), the answer to the decision problem is \YES.
Else, $e^*$ is an edge of type c). If the target configuration of $e^*$ is oriented away from the cycles, then we can output \YES.
Otherwise we will now work with $A^\circ_{\rm ini}$.
Let $e^* = (v w)$ be oriented from $v$ to $w$ in $A^\circ_{\rm ini}$.
Let $P$ be the shortest path (neglecting orientations) in $C$ from $v$ to a vertex $u$ 
of a cycle in $C$. Note that $P$ and $u$ are unique. 
Let $L \subseteq V$ be all vertices
that can be reached from $w$ by the arcs in $A^\circ_{\rm ini}$.

Rule~\ref{red-rule:twocycles} is now modified in the following way. Instead adding
a gadget to every red vertex of $C$ as before, we only add a vertex to each red vertex of $C$
that is not contained in $L$. 
Furthermore, we add a gadget to $v$ (even though it might be a blue vertex).
We then delete all edges and blue vertices of $C$ apart from $e^*$ and 
that are not contained in $L$.

Observe that, similar to Rule~\ref{red-rule:twocycles}, we have that the modified instance is a \YES
instance if and only if the old instance is a \YES instance.
After applying the modified Rule~\ref{red-rule:twocycles} until no longer possible, we obtain a 
pseudo-forest in which each blue component has at most one cycle.

Next, we consider Rule~\ref{red-rule:degree1}. If $e^*$ is not affected by the Rule, it can be applied
safely. Otherwise, if $e^*$ is the only edge adjacent to some vertex, then $e^*$ can not be reversed and
we output \NO. Hence Rule~\ref{red-rule:degree1} is safe.

Finally, we consider Rule~\ref{red-rule:2-path}. 
Again, if $e^*$ is not affected by the rule then applying the rule is safe.
Otherwise we can assume that $w$ is adjacent to precisely two blue edges, say $e^*$ and $e$. 
Similar to the original proof of Rule~\ref{red-rule:twocycles} we can show that
contracting $e$ is safe.
Hence applying any of the three (modified) rules is safe. After applying the Rules~\ref{red-rule:twocycles}-\ref{red-rule:2-path} until no longer possible, we have that the resulting instance
has at most $O(k)$ vertices and edges and thus we have a kernel of size $O(k)$.
We then have to check whether the modified start configuration is connected to a configuration
in which $e^*$ is reversed.
\end{proof}

\begin{proof}[Proof of Theorem~\ref{thm:ncl-red-edges-bound-kernel}]
  Let $(G, A_{\rm ini}, A_{\rm tar})$ be the instance of \ncl that we obtained by applying
  the reduction rules \ref{red-rule:onecycle}--\ref{red-rule:2-path}
  until no longer possible. 
  We show that $G = (V, E)$ has at most $8k$ vertices and $11k$ edges, where $k = |E^{\rm red}|$.
  Let $G_B = (V, E^{\rm blue})$ be the blue subgraph containing blue edges only.
  Furthermore, let $V_B$ be the set of blue vertices of the copies of the
  gadget shown in \figurename~\ref{fig:felix} present in $G$.

  We first bound the number of vertices and edges in $G_B - V_B$.  Note that
  each component in $G_B - V_B$ is a pseudo-forest, since otherwise
  Rule~\ref{red-rule:twocycles} is applicable.  Also note that there are at
  most $k$ blue vertices of degree 2 that are not contained in a blue cycle in
  $G_B - V_B$.  This is due to the fact that a vertex of degree 2 in $G_B -
  V_B$ is one of the exceptions of Rule \ref{red-rule:2-path}.
  All remaining vertices of
  $G_B - V_B$ have degree at least three.  Let $G'_B$ be the graph obtained
  from $G_B$ by contracting each vertex of degree 2 in $G_B - V_B$.
  
  We argue that $|V( G'_B - V_B ) | \leq 2k$ and $|E( G'_B - V_B ) | \leq 2k$.
  Since $G'_B-V_B$ is a pseudo-forest, we have that $|E( G'_B - V_B ) | \leq |V( G'_B - V_B ) |$.
  For now let us assume that $k \leq |V( G'_B - V_B ) |/2 -1$. This implies
  that there are at least $|V( G'_B - V_B ) |/2 +1$ blue vertices in $G'_B -
  V_B$. Also note that each red vertex is incident to at least one blue edge in
  $G'_B - V_B$, as otherwise it is an isolated vertex in $G'_B - V_B$.  But
  since each blue vertex has at least 3 incident edges, we have that $|E( G'_B
  - V_B ) | \geq (3 ( |V( G'_B - V_B ) |/2 +1) +  |V( G'_B - V_B ) |/2 -1)/2 >
  |V( G'_B - V_B ) |$, a contradiction. 
  Hence we have $k \geq |V( G'_B - V_B ) |/2$ and therefore $|V( G'_B - V_B ) | \leq 2k$.
  This also implies $|E( G'_B - V_B ) | \leq 2k$, as claimed. 
  Since there are up to $k$ vertices of degree 2 in $G_B - V_B$, we get $|V(
  G_B - V_B ) | \leq 3k$ and $|E( G_B - V_B ) | \leq 3k$.  For each gadget
  $G_i$ we have that $| V(G_i)| = 5$ and $|E(G_i)| = 8$.  Since $G$
  contains at most $k$ gadgets we have $|V(G)| \leq 8k$ and $|E(G)| \leq 11k$.
  
  Hence \ncl for $G$ admits a kernel of size $O(k)$\textup{;} and,
  in particular, \ncl for $G$ can be solved in time $O^*(2^{O(k)})$. 
\end{proof}

\section{Proofs Omitted from Section~\ref{sec:blue}}

\begin{proof}[Proof of Lemma~\ref{lem:blue02}]
We prove the lemma by induction on 
$|A^{\rm red}_{\rm ini} \setminus A^{\rm red}_{\rm tar}|$. 
If $\rho_{A^{\rm red}_{\rm ini}} = \rho_{A^{\rm red}_{\rm tar}}$, then the claim is obvious, because 
the sequence consisting of only one orientation $A^{\rm red}_0 = A^{\rm red}_{\rm ini}$ satisfies the conditions. 
Thus, it suffices to consider the case when $\rho_{A^{\rm red}_{\rm ini}} \not= \rho_{A^{\rm red}_{\rm tar}}$. 
In this case, there exists a vertex $u \in V$ such that $\rho_{A^{\rm red}_{\rm ini}} (u) > \rho_{A^{\rm red}_{\rm tar}}(u)$, 
because 
$\sum_{v \in V} \rho_{A^{\rm red}_{\rm ini}} (v) = \sum_{v \in V}  \rho_{A^{\rm red}_{\rm tar}}(v)$. 
Then, there exists an arc $a \in A^{\rm red}_{\rm ini} \setminus A^{\rm red}_{\rm tar}$ that enters $u$. 
Let $A^{\rm red}_1$ be the orientation of $E^{\rm red}$ obtained from $A^{\rm red}_{\rm ini}$ by reversing the direction of $a$. 
Since $|A^{\rm red}_1 \setminus A^{\rm red}_{\rm tar}|  < |A^{\rm red}_{\rm ini} \setminus A^{\rm red}_{\rm tar}|$, by induction hypothesis, 
there exists a sequence $A^{\rm red}_1, \dots , A^{\rm red}_l$ of orientations of $E^{\rm red}$
such that $\rho_{A^{\rm red}_l} = \rho_{A^{\rm red}_{\rm tar}}$, $A^{\rm red}_{i-1} \leftrightarrow A^{\rm red}_i$ for $i=2, \dots , l$, and 
$\rho_{A^{\rm red}_i}(v) \ge \min \{ \rho_{A^{\rm red}_1}(v), \rho_{A^{\rm red}_{\rm tar}}(v)\}$ for any $v \in V$ and any $i \in \{1, \dots , l\}$. 
By letting $A^{\rm red}_0 = A^{\rm red}_{\rm ini}$, the sequence $A^{\rm red}_0, A^{\rm red}_1, \dots , A^{\rm red}_l$ satisfies the conditions, 
because $A^{\rm red}_0 \leftrightarrow A^{\rm red}_1$, 
$\rho_{A^{\rm red}_1}(v) \ge \rho_{A^{\rm red}_{\rm ini}}(v)$ for each $v \in V \setminus \{u\}$, and 
$\min \{ \rho_{A^{\rm red}_1}(u), \rho_{A^{\rm red}_{\rm tar}}(u)\} = \rho_{A^{\rm red}_{\rm tar}}(u) = \min \{ \rho_{A^{\rm red}_{\rm ini}}(u), \rho_{A^{\rm red}_{\rm tar}}(u)\}$. 
\end{proof}

\begin{proof}[Proof of Lemma~\ref{lem:blue04}]
We can easily check the condition (1). 
To check the condition (2), 
we construct a digraph $\hat G=(\hat V, \hat A)$ and consider a network flow problem in it. 
Introduce a new vertex $w_e$ for each $e \in E^{\rm red}$ and two new vertices $s$ and $t$, 
and define $\hat V := V \cup \{w_e \mid e \in E^{\rm red}\} \cup \{s, t\}$. 
Define the arc set $\hat A:= \hat A_1 \cup \hat A_2 \cup \hat A_3$ by
\begin{align*}
\hat A_1 &:= \{(s, w_e) \mid e \in E^{\rm red} \}, \\
\hat A_2 &:= \{(w_e, v) \mid e \in E^{\rm red}, v\in V, \mbox{$e$ is incident to $v$ in $G$} \}, \\
\hat A_3 &:= \{(v, t) \mid v\in V \}. 
\end{align*}
For each $a \in \hat A$, define the lower bound $l(a)$ and the upper bound $u(a)$ of 
the amount of flow through $a$ as follows. 
\begin{itemize}
\item
For each $(s, w_e) \in \hat A_1$, define $l(s, w_e) := u (s, w_e) := 1$. 
\item
For each $(w_e, v) \in \hat A_2$, define $l(w_e, v) := 0$ and $u (w_e, v) := 1$. 
\item
For each $(v, t) \in \hat A_3$, define $l(v, t) := u (v, t) := d(v)$ if $v \in X$ and $d(v) \in \{0, 1\}$, 
and define $l(v, t) := 2$ and $u (v, t) := + \infty$ otherwise. 
\end{itemize}
Then, the condition (2) holds if and only if $\hat G$ has an integral $s$-$t$ flow satisfying the above constraint.  
This can be tested in polynomial time 
by a standard maximum flow algorithm (see e.g. \cite[Corollary 11.3a]{Schrijver:03}). 
\end{proof}

\begin{proof}[Proof of Lemma~\ref{lem:blue05}]
It suffices to consider the case when 
$\phi(A_{\rm ini}) \xleftrightarrow[\mathcal F]{} \phi(A_{\rm tar})$. 
Denote $\phi(A_{\rm ini}) = (A^{\rm blue}_{\rm ini}, d_{\rm ini})$ and $\phi(A_{\rm tar}) = (A^{\rm blue}_{\rm tar}, d_{\rm tar})$. 
By definition, we have either 
$d_{\rm ini} = d_{\rm tar}$ and $A^{\rm blue}_{\rm ini} \leftrightarrow A^{\rm blue}_{\rm tar}$, or 
$A^{\rm blue}_{\rm ini} = A^{\rm blue}_{\rm tar}$. 

If $d_{\rm ini} = d_{\rm tar}$ and $A^{\rm blue}_{\rm ini} \leftrightarrow A^{\rm blue}_{\rm tar}$, then 
$A_{\rm ini} \leftrightarrow A^{\rm blue}_{\rm tar} \cup A^{\rm red}_{\rm ini}$ and 
$\phi(A^{\rm blue}_{\rm tar} \cup A^{\rm red}_{\rm ini}) = (A^{\rm blue}_{\rm tar}, d_{\rm ini}) = (A^{\rm blue}_{\rm tar}, d_{\rm tar}) = \phi(A_{\rm tar})$, 
which means that $A^\circ_{\rm tar} := A^{\rm blue}_{\rm tar} \cup A^{\rm red}_{\rm ini}$ satisfies the conditions. 

Otherwise, let $A^{\rm blue} := A^{\rm blue}_{\rm ini} = A^{\rm blue}_{\rm tar}$. 
By Lemma~\ref{lem:blue02},  
we obtain a sequence $A^{\rm red}_0, A^{\rm red}_1, \dots , A^{\rm red}_l$ of orientations of $E^{\rm red}$
such that $A^{\rm red}_0 = A^{\rm red}_{\rm ini}$, $\rho_{A^{\rm red}_l} = \rho_{A^{\rm red}_{\rm tar}}$, $A^{\rm red}_{i-1} \leftrightarrow A^{\rm red}_i$ for $i=1, \dots , l$, and 
$\rho_{A^{\rm red}_i}(v) \ge \min \{ \rho_{A^{\rm red}_{\rm ini}}(v), \rho_{A^{\rm red}_{\rm tar}}(v)\}$ for any $v \in V$ and any $i \in \{0, 1, \dots , l\}$.
Then, for any $i \in \{0, 1, \dots , l\}$, we have 
\begin{align*}
2 \rho_{A^{\rm blue}} (v) + \rho_{A^{\rm red}_i}(v) 
\ge \min \{2 \rho_{A^{\rm blue}} (v) + \rho_{A^{\rm red}_{\rm ini}}(v), 2 \rho_{A^{\rm blue}} (v) + \rho_{A^{\rm red}_{\rm tar}}(v)\} 
\ge 2
\end{align*}
for any $v \in V$,    
and hence $A^{\rm blue} \cup A^{\rm red}_i$ is feasible. 
Since $A^{\rm blue} \cup A^{\rm red}_{i-1} \leftrightarrow A^{\rm blue} \cup A^{\rm red}_i$ for $i=1, \dots , l$, 
we have 
$$
(A_{\rm ini} =) A^{\rm blue} \cup A^{\rm red}_{\rm ini} \leftrightsquigarrow A^{\rm blue} \cup A^{\rm red}_l.
$$
Furthermore, 
since $\rho_{A^{\rm red}_l} = \rho_{A^{\rm red}_{\rm tar}}$, 
we have $\phi(A^{\rm blue} \cup A^{\rm red}_l) = \phi(A_{\rm tar})$. 
Therefore, $A^\circ_{\rm tar} := A^{\rm blue} \cup A^{\rm red}_l$ satisfies the conditions in the lemma. 
\end{proof}

\begin{proof}[Proof of Lemma~\ref{lem:blue06}]
We prove (i)$\Rightarrow$(ii), (ii)$\Rightarrow$(iii), and (iii)$\Rightarrow$(i), respectively. 

{\bf [(i)$\Rightarrow$(ii)]} 
If (i) holds, then $A := (A_{\rm ini} \setminus C) \cup \overline{C}$ satisfies the conditions in (ii), since it contains no arc in $C$. 

{\bf [(ii)$\Rightarrow$(iii)]} 
We prove the contraposition. 
Assume that (iii) does not hold, that is, 
there exists a vertex $u \in V(C)$ such that 
$2 \rho_{A^{\rm blue}}(u) + \rho_{A^{\rm red}}(u) = 2$
for any $A \in \mathcal A$ with $A_{\rm ini} \leftrightsquigarrow A$. 
Let $a$ be the arc in $C$ that enters $u$. 
Since we cannot reverse the direction of $a$ without violating the feasibility, 
$a$ is contained in any orientation $A \in \mathcal A$ with $A_{\rm ini} \leftrightsquigarrow A$.

{\bf [(iii)$\Rightarrow$(i)]} 
Suppose that (iii) holds. 
We take a sequence 
$A_0, A_1, \dots , A_l$ of feasible orientations of $E$ such that 
$A_0 = A_{\rm ini}$, 
$A_i$ is obtained from $A_{i-1}$ by reversing an arc $a_i \in A_{i-1}$ for $i \in \{1, 2, \dots , l\}$, and 
there exists $u \in V(C)$ such that 
$2 \rho_{A^{\rm blue}_l}(u) + \rho_{A^{\rm red}_l}(u) \ge 3$.  
By taking a minimal sequence with these conditions, we may assume that 
$a_i$ is not contained in $C$ for $i \in \{1, 2, \dots , l\}$. 
Since 
$2 \rho_{A^{\rm blue}_l}(u) + \rho_{A^{\rm red}_l}(u) \ge 3$, 
starting from $A_l$, 
we can change the direction of each arc in $C$ one by one without violating the feasibility, 
which shows that $A_l \leftrightsquigarrow (A_l \setminus C) \cup \overline{C}$. 
On the other hand, since $(A_i \setminus C) \cup \overline{C}$ is obtained from $(A_{i-1} \setminus C) \cup \overline{C}$ 
by reversing $a_i$ for $i \in \{1, 2, \dots , l\}$, we obtain
$(A_{\rm ini} \setminus C) \cup \overline{C} \leftrightsquigarrow (A_l \setminus C) \cup \overline{C}$. 
Thus, it holds that $A_{\rm ini} \leftrightsquigarrow A_l \leftrightsquigarrow (A_l \setminus C) \cup \overline{C} \leftrightsquigarrow (A_{\rm ini} \setminus C) \cup \overline{C}$. 
\end{proof}

\begin{proof}[Proof of Proposition \ref{prop:blue:A}]
If $A$ is a solution of \ProblemA, then 
$\phi_u(A) = (A^{\rm blue}, d)$ satisfies the conditions by Lemma~\ref{lem:blue18}. 
Conversely, assume that there exists a pair $(A^{\rm blue}, d)\in \mathcal{F}_u$ 
such that $2 \rho_{A^{\rm blue}} (u) + d(u) \ge 3$ and $\phi_u(A_{\rm ini}) \underset{\mathcal{F}_u}{\leftrightsquigarrow} (A^{\rm blue}, d)$.  
By Lemma~\ref{lem:blue15},  
there exists an orientation $A \in \mathcal{A}$ with $\phi_u(A) = (A^{\rm blue}, d)$ such that 
$A_{\rm ini} \leftrightsquigarrow A$. 
Since $2 \rho_{A^{\rm blue}} (u) + \rho_{A^{\rm red}} (u) \ge 2 \rho_{A^{\rm blue}} (u) + d (u) \ge 3$, 
$A$ is a solution of \ProblemA. 
\end{proof}

\begin{proof}[Proof of Corollary \ref{cor:blue01}]
Let $e^*$ be the edge of the orientation $A_{\ini}$ that we wish to reverse.
If $e^*$ is a blue edge, we simply solve the reconfiguration problem in 
$\mathcal{F}$ (see Section~\ref{sec:reconfF}). 
Since $| \mathcal{F} | = 2^{O(k)}$ this can be done in FPT time.
If $e^*$ is a red edge, we solve \ProblemA with $u$ being the head of 
$e^*$ (see Section~\ref{sec:blueprobA}).
This can also be done in FPT time.
\end{proof}

\end{document}